\newif\ifreport\reporttrue
\documentclass[onecolumn,11pt,draftcls]{IEEEtran}

\usepackage{graphicx}
\usepackage[center]{caption}
\usepackage{epsfig,latexsym,amsmath,amsfonts}
\usepackage{mathtools}

\usepackage{amssymb}
\usepackage{placeins}
\usepackage{subfigure}
\usepackage{verbatim}
\usepackage{cite,url}
\usepackage{epsf,bm}
\usepackage{color}
\usepackage{epstopdf}
\usepackage{romannum,hyperref}
\usepackage[onehalfspacing]{setspace}
\usepackage{dsfont}
\usepackage{bbm}
\usepackage{stfloats}
\usepackage[utf8]{inputenc}
\usepackage[english]{babel}
\usepackage{amsthm}
\newtheorem{definition}{Definition}
\newtheorem{remark}{Remark}

\usepackage[utf8]{inputenc}
\usepackage[english]{babel}

\newtheorem{theorem}{Theorem}
\newtheorem{lemma}[theorem]{Lemma}
\newtheorem{proposition}[theorem]{Proposition}
\newtheorem{corollary}[theorem]{Corollary}

\usepackage[lined, boxed, linesnumbered, ruled]{algorithm2e}
\IEEEoverridecommandlockouts
\begin{document}
\pagenumbering{arabic}
\title{Optimal Sampling and Scheduling for Timely Status Updates in Multi-source Networks}  



\author{\large Ahmed M. Bedewy, Yin Sun, \emph{Member, IEEE}, Sastry Kompella, \emph{Senior Member, IEEE}, and Ness B. Shroff, \emph{Fellow, IEEE}
\thanks{This paper was presented in part at ACM MobiHoc 2019 \cite{multi_source_bedewy_2}.}
\thanks {This work has been supported in part by ONR grants N00014-17-1-2417 and N00014-15-1-2166, Army Research Office grants W911NF-14-1-0368 and MURI W911NF-12-1-0385,  National Science Foundation grants CNS-1446582, CNS-1421576, CNS-1518829, and CCF-1813050, and a grant from the Defense Thrust Reduction Agency HDTRA1-14-1-0058.}
\thanks{A. M. Bedewy is with the  Department  of  ECE,  The  Ohio  State  University, Columbus, OH 43210 USA (e-mail:  bedewy.2@osu.edu).}
\thanks{Y.  Sun  is  with  the  Department  of  ECE,  Auburn  University,  Auburn,  AL 36849 USA (e-mail:  yzs0078@auburn.edu).}
\thanks{S.  Kompella  is  with  Information Technology Division, Naval Research Laboratory,  Washington, DC 20375 USA  (e-mail:  sk@ieee.org).}
\thanks{N. B.  Shroff  is  with  the  Department  of  ECE and  the  Department  of  CSE, The Ohio State University,  Columbus, OH 43210 USA  (e-mail:  shroff.11@osu.edu).}
}

%
%
\maketitle

\begin{abstract}

We consider a joint sampling and scheduling problem for optimizing data freshness in multi-source systems. Data freshness is measured by a non-decreasing penalty function of \emph{age of information}, where all sources have the same age-penalty function. Sources take turns to generate update packets, and forward them to their destinations one-by-one through a shared channel with random delay. There is a scheduler, that chooses the update order of the sources, and a sampler, that determines when a source should generate a new packet in its turn. We aim to find the optimal scheduler-sampler pairs that minimize the total-average age-penalty at delivery times (Ta-APD) and the total-average age-penalty (Ta-AP). We prove that the Maximum Age First (MAF) scheduler and the zero-wait sampler are jointly optimal for minimizing the Ta-APD. Meanwhile, the MAF scheduler and a relative value iteration with reduced complexity (RVI-RC) sampler are jointly optimal for minimizing the Ta-AP. The RVI-RC  sampler is based on a relative value iteration algorithm whose complexity is reduced by exploiting a threshold property in the optimal sampler. Finally, a low-complexity threshold-type sampler is devised via an approximate analysis of Bellman’s equation. This threshold-type sampler reduces to a simple water-filling sampler for a linear age-penalty function.
 
\end{abstract}
\newpage
\section{Introduction}\label{Int}
In recent years, significant attention has been paid to \emph{age of information} as a metric for data freshness. This is because there are a growing number of applications that require timely status updates in various networked monitoring and control systems. Examples include  sensor and environment  monitoring networks, surrounding monitoring autonomous vehicles, smart grid systems, etc. Age of information, or simply age, was introduced in \cite{adelberg1995applying,cho2000synchronizing,golab2009scheduling,KaulYatesGruteser-Infocom2012}, which is the time elapsed since the most recently received update was generated. Unlike traditional packet-based metrics, such as throughput and delay, age is a destination-based metric that captures the information lag at the destination, and is hence more apt for achieving the goal of timely updates.

There have been two major lines of research on age in single source networks: One direction is on systems with a stochastic arrival process. There are results on both queueing-based age analysis \cite{KaulYatesGruteser-Infocom2012,2012CISS-KaulYatesGruteser,CostaCodreanuEphremides_TIT,KamKompellaEphremidesTIT} and sample-path based age optimization \cite{age_optimality_multi_server,Bedewy_NBU_journal_2,multihop_optimal,bedewy2017age_multihop_journal_2}. The second direction is for the case that the packet arrival process is designable \cite{2015ISITYates,BacinogCeranUysal_Biyikoglu2015ITA,SunJournal2016,sun2018sampling_2}, where our study extends the findings in these studies to multi-source networks.

We consider random, yet discrete, transmission times such that a packet has to be processed for a random period before delivered to the destination. In practice, such random transmission times occur in many applications, such as autonomous vehicles. In particular, there are many electronic control units (ECUs) in a vehicle, that are connected to one or more sensors and actuators via a controller area network (CAN) bus \cite{ran2010design,johansson2005vehicle}. These ECUs are given different priority, based on their criticality level (e.g., ECUs in the powertrain have a higher priority compared to those connected to infotainment systems). Since high priority packets usually have hard deadlines, the transmissions of low priority packets are interrupted whenever the higher priority ones are transmitted. Therefore, information packets with lower priority see a time-varying bandwidth, and hence encounter a random transmission time. 

Another example is the wireless sensor networks that are used for environmental monitoring, human-related activities, etc. In such networks, sensor nodes may be deployed in remote areas and information is gathered from these sensors by an access point (AP) through a shared wireless channel \cite{kandris2020applications}. Since this channel is influenced by uncertain factors, the channel delay varies with time.

When the transmission time is highly random, one can observe an interesting phenomenon: it is not necessarily optimal to generate a new packet as soon as the channel becomes available. 
This phenomenon was revealed in \cite{2015ISITYates} and further explored in \cite{SunJournal2016} and \cite{sun2018sampling_2}. In the case of autonomous vehicles, many sensors may share the same CAN bus. As a result, the decision maker needs to control both the sampling times and service order of these sensors. The same observations are also applied to wireless sensor networks.

 \begin{figure}
\includegraphics[scale=0.4]{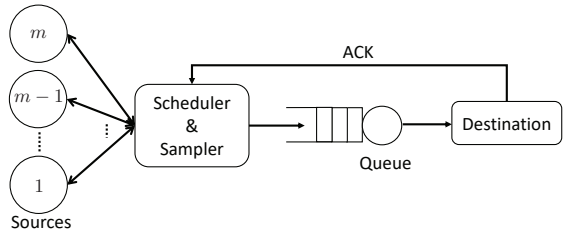}
\centering
\caption{System model}\label{Fig:sys_model}
\end{figure}
 In this paper, our goal is to investigate timely status updates in multi-source systems with random transmission times, as depicted in Fig. \ref{Fig:sys_model}. Sources take turns to generate update packets, and forward the packets to their destinations one-by-one through a shared channel with random delay. This results in a joint design problem of scheduling and sampling, where the scheduler chooses the update order of the sources, and the sampler determines when a source should generate a new packet in its turn. We find that it is optimal to first serve the source with the highest age, and, similar to the single-user case, it is not always optimal to generate packets as soon as the channel becomes available. To that end, the main contributions of this paper are outlined as follows: 
\begin{itemize}
\item We formulate the optimal scheduling and sampling  problem to optimize data freshness in single-hop, multi-source networks. We use a non-decreasing age-penalty function to represent the level of dissatisfaction of data staleness, where all sources have the same age-penalty function. We focus on minimizing the total-average age-penalty at delivery times (Ta-APD) and the total-average age-penalty (Ta-AP), where Ta-AP is more challenging to minimize. We show that our optimization problem has an important \emph{separation principle}: For any given sampler, we show that the optimal scheduling policy is the Maximum Age First (MAF) scheduler (Proposition \ref{Thm1}). Hence, the optimal scheduler-sampler pair can be obtained by fixing the scheduling policy to the MAF scheduler, and then optimize the sampler design separately.



\item We show that  the MAF scheduler and zero-wait sampler, in which a new packet is generated once the channel becomes idle, are jointly optimal for minimizing the Ta-APD (Theorem \ref{thm_tapa}). We show this result by proving the optimality of the zero-wait sampler for minimizing the Ta-APD, when the scheduling policy is fixed to the MAF scheduler. 

%

\item Interestingly, we find that zero-wait sampler does not always minimize the Ta-AP, when the MAF scheduler is employed. We show that the MAF scheduler and the relative value iteration with reduced complexity (RVI-RC) sampler are jointly optimal for minimizing the Ta-AP (Theorem \ref{thm_taa}). We take several steps to prove the optimality of the RVI-RC sampler: When the scheduling policy is fixed to the MAF scheduler, we reformulate the optimal sampling problem for minimizing the Ta-AP as an equivalent semi-Markov decision problem. We use Dynamic Programming (DP) to obtain the optimal sampler. In particular, we show that there exists a stationary deterministic sampler that can achieve optimality (Proposition \ref{thm2}). We also show that the optimal sampler has a threshold property (Proposition \ref{th_thm}), that helps in reducing the complexity of the relative value iteration (RVI) algorithm (by reducing the computations required for some system states). This results in the RVI-RC sampler in Algorithm \ref{alg1}.

\item Finally, in Section \ref{Bellman}, we devise a low-complexity threshold-type sampler via an approximate analysis of Bellman's equation whose solution is the RVI-RC sampler. In addition, for the special case of a linear age-penalty function, this threshold sampler is further simplified to the water-filling solution. The numerical results in Figs. \ref{avg_age_penalty_expo_service_prob}-\ref{avg_age_penalty_service_time}  indicate that, when the scheduler is fixed to the MAF, the performance of these approximated samplers is almost the same as that of the RVI-RC sampler. 
\end{itemize}

\section{Related Works}

Early studies have characterized the age in many interesting variants of queueing models, such as First-Come, First-Served (FCFS) \cite{2012ISIT-YatesKaul,KaulYatesGruteser-Infocom2012,KamKompellaEphremidesTIT,2015ISITHuangModiano}, Last-Come, First-Served (LCFS) with and
without preemption \cite{2012CISS-KaulYatesGruteser,RYatesTIT16}, and the queueing model with packet management \cite{CostaCodreanuEphremides_TIT, Icc2015Pappas}. The update packets in these studies arrive at the queue randomly according to a Poisson process. The work in \cite{age_optimality_multi_server,Bedewy_NBU_journal_2,multihop_optimal,bedewy2017age_multihop_journal_2} showed that Last-Generated, First-Served (LGFS)-type policies are optimal or near-optimal for minimizing a large class of age metrics in single flow multi-server and multi-hop networks.

%

Another line of research has considered the ``generate-at-will" model \cite{2015ISITYates,BacinogCeranUysal_Biyikoglu2015ITA,SunJournal2016,sun2018sampling_2}, in which the generation times (sampling times) of the update packets are controllable. The work in \cite{SunJournal2016,sun2018sampling_2} motivated the usage of nonlinear age functions from various real-time applications and designed sampling policies for optimizing nonlinear age functions in single source systems. 
Our study here extends the work in \cite{SunJournal2016,sun2018sampling_2} to a multi-source system. In this system, only one packet can be sent through the channel at a time. Therefore, a decision maker does not only consist of a sampler, but also a scheduler, which makes the problem even more challenging.

The scheduling problem for multi-source networks with different scenarios was considered in \cite{bedewy2019optimal,aphermedis_he2017optimal,kadota2016minimizing,kadota2016minimizing_journal,hsu2017scheduling,kadota2018optimizing,hsu2018age,li2013throughput,Yin_multiple_flows,yates2017status,talak2018distributed,talak2018optimizing,talak2018scheduling,talak2018optimizing2}.  In \cite{aphermedis_he2017optimal}, the authors found that the scheduling problem for minimizing the age in wireless networks  under physical interference constraints is NP-hard. Optimal scheduling for age minimization in a broadcast network was studied in \cite{kadota2016minimizing,kadota2016minimizing_journal,hsu2017scheduling,kadota2018optimizing,hsu2018age}, where a single source can be scheduled at a time. In addition, it was found that a maximum age first (MAF) service discipline is useful for reducing the age in various multi-source systems with different service time distributions in \cite{Yin_multiple_flows,li2013throughput,hsu2017scheduling,kadota2016minimizing,kadota2016minimizing_journal}. In contrast to our study, the generation of the update packets in \cite{Yin_multiple_flows,li2013throughput,aphermedis_he2017optimal,kadota2016minimizing,kadota2016minimizing_journal,hsu2017scheduling,kadota2018optimizing,hsu2018age} is uncontrollable and they arrive randomly at the transmitter. Age analysis of the status updates over  a multiaccess channel was considered in \cite{yates2017status}. The studies in \cite{talak2018distributed,talak2018optimizing,talak2018scheduling,talak2018optimizing2} considered the age optimization problem in a wireless network with general interference constraints and channel uncertainty. Our result in  Corollary \ref{zero_opt_time_slotted} suggests that if the packet transmission time is fixed as in time-slotted systems \cite{aphermedis_he2017optimal,kadota2016minimizing,kadota2016minimizing_journal,hsu2017scheduling,kadota2018optimizing,hsu2018age,li2013throughput,yates2017status,talak2018distributed,talak2018optimizing,talak2018scheduling,talak2018optimizing2}, then it is optimal to sample as soon as the channel becomes available. However, this is not necessarily true otherwise.

\section{Model and Formulation}\label{sysmod}
\subsection{Notations}
 We use $\mathbb{N}^+$ to represent the set of non-negative integers, $\mathbb{R}^+$ is the set of non-negative real numbers, $\mathbb{R}$ is the set of real numbers, and $\mathbb{R}^n$ is the set of $n$-dimensional real Euclidean space.  We use $t^-$ to denote the time instant just before $t$. Let $\mathbf{x}=(x_1,x_2,\ldots,x_n)$ and $\mathbf{y}=(y_1,y_2,\ldots,y_n)$ be two vectors in $\mathbb{R}^n$, then we denote $\mathbf{x}\leq\mathbf{y}$ if $x_i\leq y_i$ for $i=1,2,\ldots,n$. 
Also, we use $x_{[i]}$ to denote the $i$-th largest component of vector $\mathbf{x}$.

\subsection{System Model}
We consider a status update system with $m$ sources as shown in Fig. \ref{Fig:sys_model}, where each source observes a time-varying process. An update packet is generated from a source and is then sent over an error-free delay channel to the destination, where only one packet can be sent at a time. A decision maker controls the transmission order of the sources and the generation times of the update packets for each source. This is known as the ``generate-at-will'' model \cite{BacinogCeranUysal_Biyikoglu2015ITA,2015ISITYates,SunJournal2016} (i.e., the update packets can be generated at any time). 

We use $S_i$ to denote the generation time of the $i$-th generated packet from all sources, called packet $i$. Moreover, we use $r_i$ to represent the source index from which packet $i$ is generated. The channel is modeled as an FCFS queue with random \emph{i.i.d.} service time $Y_i$, where $Y_i$ represents the service time of packet $i$, $Y_i\in\mathcal{Y}$, and $\mathcal{Y}\subset\mathbb{R}^+$ is a finite and bounded set. We also assume that $0<\mathbb{E}[Y_i]<\infty$ for all $i$. We suppose that the decision maker knows the idle/busy state of the server through acknowledgments (ACKs) from the destination with zero delay. If an update packet is generated while the server is busy, this packet needs to wait in the queue until its transmission opportunity, and becomes stale while waiting. Hence, there is no loss of optimality to avoid generating an update packet during the busy periods. As a result,  a packet is served immediately once it is generated. Let $D_i$ denote the delivery time of packet $i$, where $D_i=S_i+Y_i$. After the delivery of packet $i$ at time $D_i$, the decision maker may insert a waiting time $Z_i$ before generating a new packet (hence, $S_{i+1}=D_i+Z_i$)\footnote{We suppose that $D_0=0$. Thus, we have $S_1=Z_0$.}, where $Z_i\in\mathcal{Z}$, and $\mathcal{Z}\subset\mathbb{R}^+$ is a finite and bounded set\footnote{We suppose that we always have $0\in\mathcal{Z}$.}.


\begin{figure}
\includegraphics[scale=0.25]{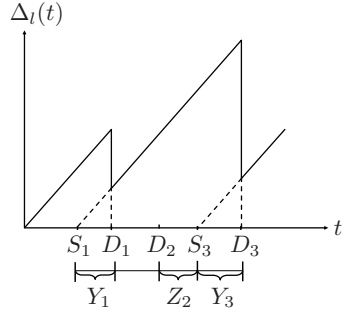}
\centering
\caption{The age $\Delta_l(t)$ of source $l$, where we suppose that the first and third packets are generated from source $l$, i.e., $r_1=r_3=l$.}\label{age_proc}
\end{figure}
 At any time $t$, the most recently delivered packet from source $l$ is generated at time
\begin{equation}
U_l(t)=\max\{S_i: r_i=l, D_i\leq t\}.
\end{equation}
\emph{Age of information}, or simply the \emph{age}, for source $l$ is defined as  \cite{adelberg1995applying,cho2000synchronizing,golab2009scheduling,KaulYatesGruteser-Infocom2012} 
\begin{equation}
\Delta_l(t)=t-U_l(t).
\end{equation} 
As shown in Fig. \ref{age_proc}, the age increases linearly with $t$ but is reset to a smaller value with the delivery of a fresher packet. We suppose that the age $\Delta_l(t)$ is right-continuous. The age process for source $l$ is given by $\{\Delta_l(t), t\geq 0\}$. We suppose that the initial age values $\Delta_l(0^-)$ for all $l$ are known to the system. For notation simplicity, we use $a_{li}$ to denote the age value of source $l$ at time $D_i$, i.e., $a_{li}=\Delta_l(D_i)$\footnote{Since the age process is right-continuous, if packet $i$ is delivered from source $l$, then $\Delta_l(D_i)$ is the age value of source $l$ just after the delivery time $D_i$.}

 For each source $l$, we consider an age-penalty function $g(\Delta_l(t))$ of the age $\Delta_l(t)$. The function $g : [0,\infty)\to \mathbb{R}$ is non-decreasing and is not necessarily convex or continuous. We suppose that $\mathbb{E}[\vert\int_{a}^{a+x}g(\tau)d\tau\vert]<\infty$ whenever $x<\infty$.  It was recently shown in \cite{sun2018sampling_2} that, under certain conditions, information freshness metrics expressed in terms of auto-correlation functions, the estimation error of signal values, and mutual information, are monotonic functions of the age. Moreover, the age-penalty function $g(\cdot)$ can be used to represent the level of dissatisfaction of data staleness in different applications based on their demands. For instance, a stair-shape function $g(x)=\lfloor x\rfloor$ can be used to characterize the dissatisfaction for data staleness when the information of interest is checked periodically,  an exponential function $g(x)=e^x$ can be utilized in online learning and control applications in which the demand for updating data increases quickly with age, and an indicator function $g(x)=\mathds{1}(x>q)$ can be used to indicate the dissatisfaction of the violation of an age limit $q$.

\subsection{Decision Policies}
A decision policy, denoted by $d$, controls the following: i) the scheduler, denoted by $\pi$, that determines the source to be served at each transmission opportunity  $\pi\triangleq (r_1, r_2, \ldots)$, ii) the sampler, denoted by $f$, that determines the packet generation times $f\triangleq (S_1, S_2, \ldots)$, or equivalently, the sequence of waiting times $f\triangleq (Z_0, Z_1, \ldots)$. Hence, $d=(\pi,f)$ implies that a decision policy $d$ employs the scheduler $\pi$ and the sampler $f$. Let $\mathcal{D}$ denote the set of causal decision policies in which decisions are made based on the history and current information of the system. Observe that  $\mathcal{D}$ consists of  $\Pi$ and $\mathcal{F}$, where $\Pi$ and $\mathcal{F}$ are the sets of causal schedulers and samplers, respectively.  

%

After each delivery, the decision maker chooses the source to be served, and imposes a waiting time before the generation of the new packet. Next, we present our optimization problems.  
\subsection{Optimization Problem}
 We define two metrics to assess the long term age performance over our status update system in \eqref{peak_age_def} and \eqref{avg_age_def}. Consider the time interval $[0, D_n]$. For any decision policy $d=(\pi,f)$, we define the total-average age-penalty at delivery times (Ta-APD) as
\begin{equation}\label{peak_age_def}
\Delta_{\text{avg-D}}(\pi,f)=\limsup_{n\rightarrow\infty}\frac{1}{n}\mathbb{E}\left[\sum_{l=1}^m\sum_{i=1}^{n}g\left(\Delta_{l}(D_i^{-})\right)\right],
\end{equation}
and the total-average age-penalty per unit time (Ta-AP) as
\begin{equation}\label{avg_age_def}
\Delta_{\text{avg}}(\pi,f)=\limsup_{n\rightarrow\infty}\frac{\mathbb{E}\left[\sum_{l=1}^{m}\int_{0}^{D_n}g\left(\Delta_l(t)\right)dt\right]}{\mathbb{E}\left[D_n\right]}.
\end{equation}
In this paper, we aim to minimize both the Ta-APD and the Ta-AP separately. In other words, we seek a decision policy $d=(\pi,f)$ that solves the following optimization problems:
\begin{equation}\label{optimal_eq_p}
\bar{\Delta}_{\text{avg-D-opt}}\triangleq\min_{\pi \in \Pi, f\in\mathcal{F}}\Delta_{\text{avg-D}}(\pi,f),
\end{equation} 
and
\begin{equation}\label{optimal_eq}
\bar{\Delta}_{\text{avg-opt}}\triangleq\min_{\pi \in \Pi, f\in\mathcal{F}}\Delta_{\text{avg}}(\pi,f),
\end{equation}
where $\bar{\Delta}_{\text{avg-D-opt}}$ and $\bar{\Delta}_{\text{avg-opt}}$  are the optimum objective values of Problems \eqref{optimal_eq_p} and \eqref{optimal_eq}, respectively. Due to the large decision policy space, the optimization problem is quite challenging. In particular, we need to seek the optimal decision policy that controls both the scheduler and sampler to minimize the Ta-APD and the Ta-AP. In the next section, we discuss our approach to tackle these optimization problems.

\section{Optimal Decision Policy}\label{optimal_policies}
We first show that our optimization problems in \eqref{optimal_eq_p} and \eqref{optimal_eq} have  an important separation principle: Given the generation times of the update packets, the Maximum Age First (MAF) scheduler provides the best age performance compared to any other scheduler. What remains to be addressed is the question of finding the best sampler that solves Problems  \eqref{optimal_eq_p} and \eqref{optimal_eq}, given that the scheduler is fixed to the MAF. Next, we present our approach to solve our optimization problems in detail. 
\subsection{Optimal Scheduler}
We start by defining the MAF scheduler as follows:
\begin{definition}[\cite{li2013throughput,hsu2017scheduling,kadota2016minimizing,kadota2016minimizing_journal,Yin_multiple_flows}]
Maximum Age First scheduler: In this scheduler, the source with the maximum age is served first among all sources. Ties are broken arbitrarily.
\end{definition}
 For simplicity, let $\pi_{\text{MAF}}$ represent the MAF scheduler. The age performance of $\pi_{\text{MAF}}$ scheduler is characterized in the following proposition:
\begin{proposition}\label{Thm1}
For all $f\in\mathcal{F}$
\begin{equation}\label{Thm1eq1b}
\Delta_{\text{\emph{avg-D}}}(\pi_{\text{MAF}},f)=\min_{\pi\in\Pi} \Delta_{\text{\emph{avg-D}}}(\pi,f),
\end{equation}
\begin{equation}\label{Thm1eq2b}
\Delta_{\text{\emph{avg}}}(\pi_{\text{MAF}},f)=\min_{\pi\in\Pi} \Delta_{\text{\emph{avg}}}(\pi,f).
\end{equation}
That is, the MAF scheduler minimizes both the Ta-APD and the Ta-AP in \eqref{peak_age_def} and \eqref{avg_age_def} among all schedulers in $\Pi$. 
\end{proposition}
\begin{proof}
One of the key ideas of the proof is as follows: Given any sampler, that controls the generation times of the update packets, we only control from which source a packet is generated. We couple the policies such that the packet delivery times are fixed under all decision policies. In the MAF scheduler, a source with maximum age becomes the source with minimum age among the $m$ sources after each delivery. Under any arbitrary scheduler, a packet can be generated from any source, which is not necessarily the one with the maximum age, and the chosen source becomes the one with minimum age among the $m$ sources after the delivery. Since the age-penalty function, $g(\cdot)$, is non-decreasing, the MAF scheduler provides a better age performance compared to any other scheduler. For details, see Appendix~\ref{Appendix_A}.
\end{proof}
Proposition \ref{Thm1} is proven by using a sample-path proof technique that was recently developed in \cite{Yin_multiple_flows}. The difference is that the authors  in \cite{Yin_multiple_flows} proved the results for symmetrical packet generation and arrival processes, while  we consider here that the packet generation times are controllable. 
We found that the same proof technique applies to both cases. Observe that, Proposition 1 holds when all sources are equally prioritized. However, for the sources with different priorities (i.e. different age-penalty functions), this result does not hold anymore. This is because the order of the age-penalty values of various sources may change with time.


Proposition \ref{Thm1} 
helps us conclude the separation principle that the optimal sampler can be optimized separately, given that the scheduling policy is fixed to the MAF scheduler. Hence, the optimization problems \eqref{optimal_eq_p} and \eqref{optimal_eq} reduce to the following:
\begin{equation}\label{optimal_eq2'}
\bar{\Delta}_{\text{avg-D-opt}}\triangleq\min_{f\in\mathcal{F}}\Delta_{\text{avg-D}}(\pi_{\text{MAF}},f),
\end{equation}
\begin{equation}\label{optimal_eq2}
\bar{\Delta}_{\text{avg-opt}}\triangleq\min_{f\in\mathcal{F}}\Delta_{\text{avg}}(\pi_{\text{MAF}},f).
\end{equation}

\begin{figure*}[!tbp]
 \centering
 \subfigure[The age evolution of source 1.]{
  \includegraphics[scale=0.25]{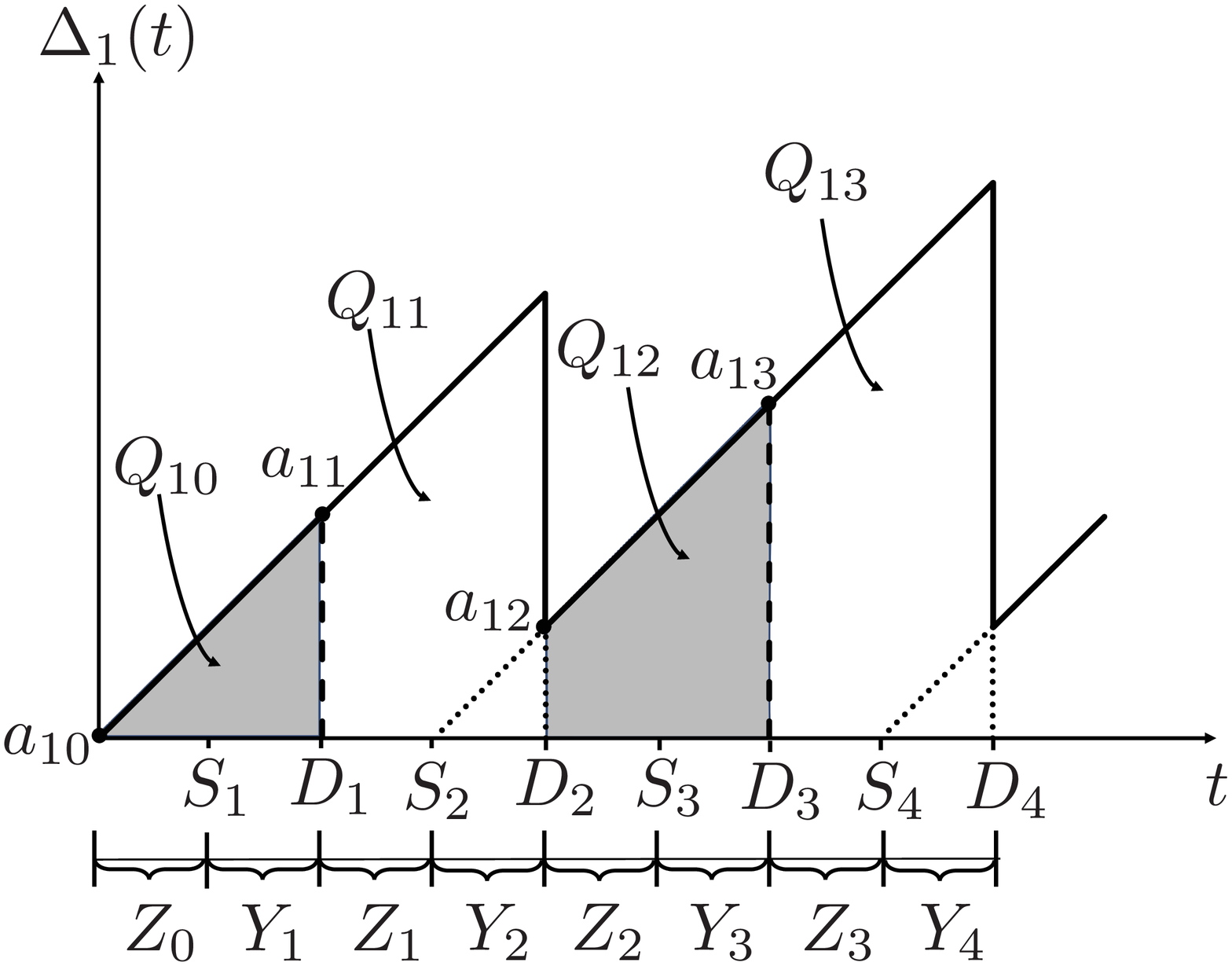}
   \label{a}
   }
 \subfigure[The age evolution of source 2.]{
  \includegraphics[scale=0.25]{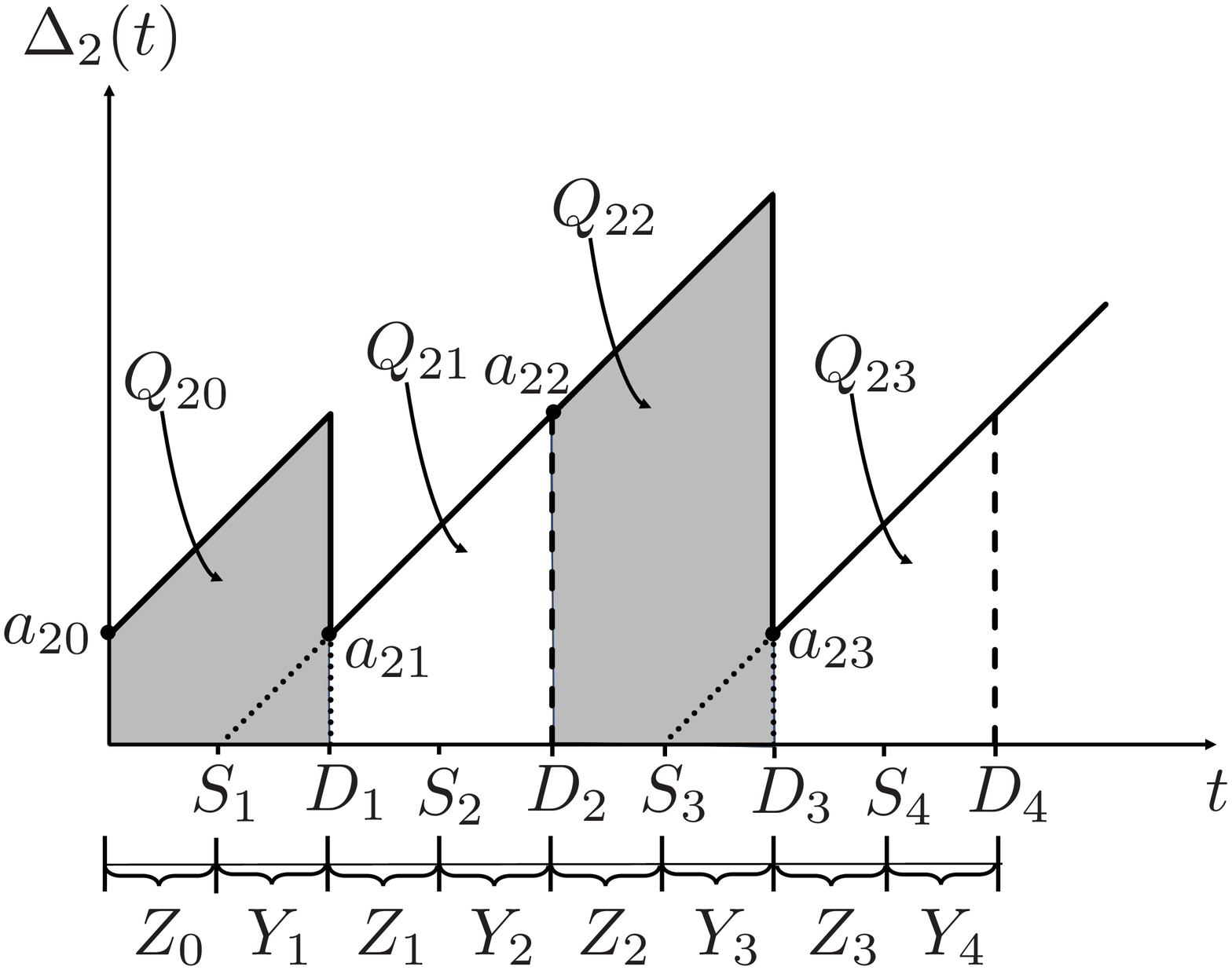}
   \label{b}
   }
   \caption{The age processes evolution of the MAF scheduler in a two-sources information update system. Source 2 has a higher initial age than Source 1. Thus, Source 2 starts service and Packet 1 is generated from Source 2, which is delivered at time $D_1$. Then, Source 1 is served and Packet 2 is generated from Source 1, which is delivered at time $D_2$. The same operation is repeated over time.}
   \label{fig:ages_evolv}
\end{figure*}
By fixing the scheduling policy to the MAF scheduler, the evolution of the age processes of the sources is as follows: The sampler may impose a waiting time $Z_i$ before generating packet $i+1$ at time $S_{i+1}=D_i+Z_i$ from the source with the maximum age at time $t=D_i$. Packet $i+1$ is delivered at time $D_{i+1}=S_{i+1}+Y_{i+1}$ and the age of the source with maximum age drops to the minimum age with the value of $Y_{i+1}$, while the age processes of other sources increase linearly with time without change. This operation is repeated with time and the age processes evolve accordingly. An example of age processes evolution is shown in Fig. \ref{fig:ages_evolv}. Next, we seek the optimal sampler for Problems \eqref{optimal_eq2'} and \eqref{optimal_eq2}. 
\subsection{Optimal Sampler for Problem \eqref{optimal_eq2'}}
Now, we show that the MAF scheduler and the zero-wait sampler are jointly optimal for minimizing the Ta-APD as follows:

\begin{theorem}\label{thm_tapa}
The MAF scheduler and the zero-wait sampler form an optimal solution for Problem \eqref{optimal_eq_p}.
\end{theorem}
\begin{proof}
We prove Theorem \ref{thm_tapa} by proving that the zero-wait sampler is optimal for Problem \eqref{optimal_eq2'}. In particular, we show that the Ta-APD is an increasing function of the packets waiting times $Z_i$'s. For details, see Appendix \ref{Appendix_A'}.
\end{proof}

\begin{remark}
The results in Proposition \ref{Thm1} and Theorem \ref{thm_tapa} hold even if $\mathcal{Y}$ and $\mathcal{Z}$ are unbounded and uncountable sets. Indeed, the finiteness assumption of $\mathcal{Y}$ and $\mathcal{Z}$ is only needed for the utilization of the DP technique in the next subsection. 
\end{remark}

\subsection{Optimal Sampler for Problem \eqref{optimal_eq2}}
Although the zero-wait sampler is the optimal sampler for minimizing the Ta-APD, it is not clear whether it also minimizes the Ta-AP. This is because the latter metric may not be a non-decreasing function of the waiting times as we will see later, which makes Problem \eqref{optimal_eq2} more challenging. 
 Next, we derive the Ta-AP when the MAF scheduler is employed and reformulate Problem \eqref{optimal_eq2} as a semi-Markov decision problem.
\subsubsection{Reformulation of Problem \eqref{optimal_eq2}}
We start by analyzing  the Ta-AP when the scheduling policy is fixed to the MAF scheduler.  We decompose the area under each curve $g(\Delta_l(t))$ into a sum of disjoint geometric parts. Observing Fig. \ref{fig:ages_evolv} \footnote{Observe that a special age-penalty function is depicted in Fig. \ref{fig:ages_evolv}, where we choose $g(x)=x$ to simplify the illustration.}, this area in the time interval $[0, D_n]$, where $D_n=\sum_{i=0}^{n-1}Z_i+Y_{i+1}$, can be seen as the concatenation of the areas $Q_{li}$, $0\leq i\leq n-1$. Thus, we have
 \begin{equation}\label{integral_eq_1'}
 \int_0^{D_n}g(\Delta_l(t))dt=\sum_{i=0}^{n-1}Q_{li},
 \end{equation}
 where
 \begin{align}\label{eq_taa_p_13}
 Q_{li}=\int_{D_i}^{D_{i+1}}g(\Delta_l(t))dt=\int_{D_i}^{D_i+Z_i+Y_{i+1}}g(\Delta_l(t))dt.
 \end{align}
For $t\in[D_i,D_{i+1})$, we have 
\begin{align}
\Delta_l(t)=t-U_l(t)=t-(D_i-a_{li}),
\end{align}
where $(D_i-a_{li})$ represents the generation time of the last delivered packet from source $l$ before time $D_{i+1}$. By performing a change of variable in \eqref{eq_taa_p_13}, we get
\begin{align}
Q_{li}=\int_{a_{li}}^{a_{li}+Z_i+Y_{i+1}}g(\tau)d\tau.
\end{align}
Hence, the Ta-AP can be rewritten as
\begin{equation}\label{total_avg_age}
\limsup_{n\rightarrow\infty}\frac{\sum_{i=0}^{n-1}\mathbb{E}\left[\sum_{l=1}^m\int_{a_{li}}^{a_{li}+Z_i+Y_{i+1}}g(\tau)d\tau\right]}{\sum_{i=0}^{n-1}\mathbb{E}\left[Z_i+Y_{i+1}\right]}.
\end{equation}
Using this, the optimal sampling problem for minimizing the Ta-AP,  given that the scheduling policy is fixed to the MAF scheduler, can be cast as
\begin{equation}\label{optimal_eq_sampler}
\bar{\Delta}_{\text{avg-opt}}\triangleq\min_{f\in\mathcal{F}}
\limsup_{n\rightarrow\infty}\frac{\sum_{i=0}^{n-1}\mathbb{E}\left[\sum_{l=1}^m\int_{a_{li}}^{a_{li}+Z_i+Y_{i+1}}g(\tau)d\tau\right]}{\sum_{i=0}^{n-1}\mathbb{E}\left[Z_i+Y_{i+1}\right]}.
\end{equation}
Since $\vert\int_{a_{li}}^{a_{li}+Z_i+Y_{i+1}}g(\tau)d\tau\vert<\infty$ for all $Z_i\in\mathcal{Z}$ and $Y_{i}\in\mathcal{Y}$, and  $\mathbb{E}[Y_i]>0$ for all $i$, $\bar{\Delta}_{\text{avg-opt}}$ is bounded. Note that Problem \eqref{optimal_eq_sampler} is hard to solve in the current form. Therefore, we reformulate it. We consider the following optimization problem with a parameter $\beta\geq 0$:
\begin{equation}\label{equivilent_optimal_sampler}
\begin{split}
\Theta(\beta)\triangleq\min_{f\in\mathcal{F}}\limsup_{n\rightarrow\infty}\frac{1}{n}\sum_{i=0}^{n-1}\mathbb{E}\Bigg[\sum_{l=1}^m\int_{a_{li}}^{a_{li}+Z_i+Y_{i+1}}g(\tau)d\tau-\beta(Z_i+Y_{i+1})\Bigg],
\end{split}
\end{equation}
where $\Theta\left(\beta\right)$ is the optimal value of \eqref{equivilent_optimal_sampler}.
\begin{lemma}\label{pro_simple}
The following assertions are true:
\begin{itemize}
\item[(i).] $\bar{\Delta}_{\text{avg-opt}}\lesseqqgtr\beta$ if and only if $\Theta(\beta)\lesseqqgtr 0$.
\item[(ii).] If $\Theta(\beta)=0$, then the optimal sampling policies that solve \eqref{optimal_eq_sampler} and \eqref{equivilent_optimal_sampler} are identical.
\end{itemize}
\end{lemma}
\ifreport 
\begin{proof}
The proof of Lemma \ref{pro_simple} is similar to the proof of \cite[Lemma 2]{Sun_reportISIT17}. The difference is that a regenerative assumption of the inter-sampling times is used to prove the result in \cite{Sun_reportISIT17}; instead, we use the boundedness of the inter-sampling times to prove the result. For the sake of completeness, we modify the proof accordingly and provide it in  Appendix~\ref{Appendix_B}.
\end{proof}
\fi

As a result of Lemma \ref{pro_simple}, the solution to \eqref{optimal_eq_sampler} can be obtained by solving \eqref{equivilent_optimal_sampler} and seeking a $\beta=\bar{\Delta}_{\text{avg-opt}}\geq 0$ such that $\Theta(\bar{\Delta}_{\text{avg-opt}})=0$. Lemma \ref{pro_simple} helps us to utilize the DP technique to obtain the optimal sampler. Note that without Lemma \ref{pro_simple}, it would be quite difficult to use the DP technique to solve \eqref{optimal_eq_sampler} optimally. Next, we illustrate our solution approach to  Problem \eqref{equivilent_optimal_sampler} in detail.
\subsubsection{The solution of Problem \eqref{equivilent_optimal_sampler}}
Following the methodology proposed in \cite{Bertsekas1996bookDPVol2}, 
when $\beta=\bar{\Delta}_{\text{avg-opt}}$, Problem \eqref{equivilent_optimal_sampler} is equivalent to an average cost per stage problem. According to \cite{Bertsekas1996bookDPVol2}, we describe the components of this problem in detail below.
\begin{itemize}
\item \textbf{States:} At stage\footnote{From henceforward, we assume that the duration of stage $i$ is $[D_i,D_{i+1})$.} $i$, the system state is specified by
\begin{equation}
\mathbf{s}(i)=(a_{[1]i}, \ldots, a_{[m]i}),
\end{equation}
where $a_{[l]i}$ is the $l$-th largest age of the sources at stage $i$, i.e., it is the $l$-th largest component of the vector $(a_{1i},\ldots,a_{mi})$. We use $\mathcal{S}$ to denote the state-space including all possible states. Notice that $\mathcal{S}$ is finite and bounded because $\mathcal{Z}$ and $\mathcal{Y}$ are finite and bounded. 


\item \textbf{Control action:} At stage $i$, the action that is taken by the sampler is $Z_i\in\mathcal{Z}$. 

\item \textbf{Random disturbance:} In our model, the random disturbance occurring at stage $i$ is $Y_{i+1}$, which is independent of the system state and the control action. 

\item \textbf{Transition probabilities:} If the control $Z_i=z$ is applied at stage $i$ and the service time of packet $i+1$ is $Y_{i+1}=y$, then the evolution of the system state from $\mathbf{s}(i)$ to $\mathbf{s}(i+1)$ is as follows:
\begin{equation}\label{state_evol}
\begin{split}
&a_{[m]i+1}=y,\\
&a_{[l]i+1}=a_{[l+1]i}+z+y,~ l=1,\ldots,m-1.
\end{split}
\end{equation}
We let $\mathbb{P}_{\mathbf{s}\mathbf{s}'}(z)$ denote the transition probabilities
\begin{equation}
 \mathbb{P}_{\mathbf{s}\mathbf{s}'}(z)\!=\!\mathbb{P}(\mathbf{s}(i\!+\!1)\!=\!\mathbf{s}'\vert \mathbf{s}(i)\!=\!\mathbf{s}, Z_i\!=\!z),~\mathbf{s},\mathbf{s}'\!\in\!\mathcal{S}. 
\end{equation}
When $\mathbf{s}=(a_{[1]},\ldots,a_{[m]})$ and $\mathbf{s}'=(a'_{[1]},\ldots,a'_{[m]})$, the law of the  transition probability is given by
\begin{equation}\label{trans_prob_eq}
\mathbb{P}_{\mathbf{s}\mathbf{s}'}(z)=\left\{ \begin{array}{cl}
\mathbb{P}(Y_{i+1}=y) \!\!\!\!\!\!& \  \ \text{if}~ a'_{[m]}\!=\!y~\text{and}\\ & \  \ a'_{[l]}\!=\!a_{[l\!+\!1]}\!+\!z\!+\!y~\text{for}~l\!\neq\! m; \\
0 &~ \  \text{else.} \end{array} \right. 
\end{equation}
\item\textbf{Cost Function:} Each time the system is in stage $i$ and control $Z_i$ is applied, we incur a cost 
\begin{equation}
\begin{split}
C(\mathbf{s}(i), Z_i, Y_{i+1})=\sum_{l=1}^m\int_{a_{[l]i}}^{a_{[l]i}+Z_i+Y_{i+1}}g(\tau)d\tau-\bar{\Delta}_{\text{avg-opt}}(Z_i+Y_{i+1}).
\end{split}
\end{equation}
 To simplify notation, we use the expected cost $C(\mathbf{s}(i), Z_i)$ as the cost per stage, i.e., 
\begin{equation}\label{expected_cost_not_detail}
C(\mathbf{s}(i), Z_i)=\mathbb{E}_{Y_{i+1}}\left[C(\mathbf{s}(i), Z_i, Y_{i+1})\right], 
\end{equation} 
where $\mathbb{E}_{Y_{i+1}}$ is the expectation with respect to $Y_{i+1}$, which is independent of $\mathbf{s}(i)$ and $Z_i$. 
It is important to note that there exists $c\in\mathbb{R}^+$ such that $\vert C(\mathbf{s}(i), Z_i)\vert\leq c$ for all $\mathbf{s}(i)\in\mathcal{S}$ and $Z_i\in\mathcal{Z}$. This is because $\mathcal{Z}$, $\mathcal{Y}$, $\mathcal{S}$, and $\bar{\Delta}_{\text{avg-opt}}$ are bounded. 
\end{itemize}
In general, the average cost per stage under a sampling policy $f\in\mathcal{F}$ is given by 
\begin{equation}\label{avg_cost_per_stage_1'}
\limsup_{n\rightarrow\infty}\frac{1}{n}\mathbb{E}\left[ \sum_{i=0}^{n-1}C(\mathbf{s}(i), Z_i)\right].
\end{equation}
We say that a sampling policy $f\in\mathcal{F}$ is \emph{average-optimal} if it minimizes the average cost per stage in \eqref{avg_cost_per_stage_1'}.
Our objective is to find the average-optimal sampling policy. A policy $f$ is called 
 a stationary deterministic policy if $Z_i=q(\mathbf{s}(i))$ for all $i=0,1,\ldots$, where $q:\mathbb{R}^{m+}\to \mathcal{Z}$ is a deterministic function.
 In the next proposition, we show that there is  a stationary deterministic policy that is average-optimal. 
\begin{proposition}\label{thm2}
There exist a scalar $\lambda$ and a function $h$ that satisfy the following Bellman's equation:
\begin{equation}\label{bell1'}
\lambda+h(\mathbf{s})=\min_{z\in\mathcal{Z}}\left( C(\mathbf{s},z)+\sum_{\mathbf{s}'\in\mathcal{S}}\mathbb{P}_{\mathbf{s}\mathbf{s}'}(z)h(\mathbf{s}')\right),
\end{equation}
where $\lambda$ is the optimal average cost per stage that is independent of the initial state $\mathbf{s}(0)$ and satisfies 
\begin{equation}
\lambda=\lim_{\alpha\rightarrow 1}(1-\alpha)J_{\alpha}(\mathbf{s}), \forall \mathbf{s}\in\mathcal{S}, 
\end{equation}
and $h(\mathbf{s})$ is the relative cost function that, for any state $\mathbf{o}$, satisfies
\begin{equation}\label{relative_cost_eq}
h(\mathbf{s})=\lim_{\alpha\rightarrow 1}(J_\alpha(\mathbf{s})-J_\alpha(\mathbf{o})), \forall \mathbf{s}\in\mathcal{S},
\end{equation}
where $J_\alpha(\mathbf{s})$ is the optimal total expected $\alpha$-discounted cost function, which is defined by
\begin{equation}\label{j_alpha}
 J_\alpha(\mathbf{s})=\min_{f\in\mathcal{F}}\limsup_{n\rightarrow\infty}\mathbb{E}\left[ \sum_{i=0}^{n-1}\alpha^iC(\mathbf{s}(i), Z_i)\right],\mathbf{s}(0)=\mathbf{s}\in\mathcal{S},
\end{equation} 
where $0<\alpha<1$ is the discount factor. Furthermore, there exists a stationary deterministic policy that attains the minimum in \eqref{bell1'} for each $\mathbf{s}\in\mathcal{S}$ and is average-optimal.
\end{proposition}  
\begin{proof}
According to \cite[Proposition 4.2.1 and Proposition 4.2.6]{Bertsekas1996bookDPVol2}, it is enough to show that for every two states $\mathbf{s}$ and $\mathbf{s}'$, there exists a stationary deterministic policy $f$ such that for some $k$, we have $\mathbb{P}\left[\mathbf{s}(k)=\mathbf{s}'\vert\mathbf{s}(0)=\mathbf{s}, f\right] > 0$, i.e., we have a communicating Markov decision process (MDP). Observe that the proof idea of this proposition is different from those used in literature such as \cite{hsu2017scheduling,hsu2018age}, where they have used the discounted cost problem to show their results and then connect it to the average cost problem. For details, see Appendix \ref{Appendix_C}.
\end{proof}
We can deduce from Proposition \ref{thm2} that the optimal waiting time is a fixed function of the state $\mathbf{s}$. Next, we use the RVI algorithm to obtain the optimal sampler for minimizing the Ta-AP, and then exploit the structure of our problem to reduce its complexity.



\textbf{Optimal Sampler Structure:}
The RVI algorithm \cite[Section 9.5.3]{puterman2005markov}, \cite[Page 171]{kaelbling1996recent} can be used to solve Bellman's equation \eqref{bell1'}. Starting with an arbitrary state $\mathbf{o}$, a single iteration for the RVI algorithm is given as follows:
\begin{equation}\label{RVI1}
\begin{split}
Q_{n+1}(\mathbf{s},z)&= C(\mathbf{s},z)+\sum_{\mathbf{s}'\in\mathcal{S}}\mathbb{P}_{\mathbf{s}\mathbf{s}'}(z)h_n(\mathbf{s}'),\\
J_{n+1}(\mathbf{s})&=\min_{z\in\mathcal{Z}}(Q_{n+1}(\mathbf{s},z)),\\
h_{n+1}(\mathbf{s})&= J_{n+1}(\mathbf{s})-J_{n+1}(\mathbf{o}),
\end{split}
\end{equation}
where $Q_{n+1}(\mathbf{s},z)$, $J_{n}(\mathbf{s})$, and $h_{n}(\mathbf{s})$ denote the state action value function, value function, and relative value function for iteration $n$, respectively. In the beginning, we set $J_{0}(\mathbf{s})=0$ for all $\mathbf{s}\in\mathcal{S}$, and then we repeat the iteration of the RVI algorithm as described before\footnote{ According to \cite{puterman2005markov, kaelbling1996recent}, a sufficient condition for the convergence of the RVI algorithm is the aperiodicity of the transition matrices of stationary deterministic optimal policies. In our case, these transition matrices depend on the service times. This condition can always be achieved by applying the aperiodicity transformation as explained in \cite[Section 8.5.4]{puterman2005markov}, which is a simple transformation. However, This is not always necessary to be done.}.
 
 The complexity of the RVI algorithm is high due to many sources (i.e., the curse of dimensionality \cite{powell2007approximate}). Thus, we need to simplify the RVI algorithm. To that end, we show that the optimal sampler has a threshold property that can reduce the complexity of the RVI algorithm. Define $z^\star_s$ as the optimal waiting time for state $\mathbf{s}$, and $Y$ as a random variable that has the same distribution as $Y_i$. The threshold property in the optimal sampler is manifested in the following proposition:

\begin{proposition}\label{th_thm}
If the state $\mathbf{s}=(a_{[1]}, \ldots, a_{[m]})$ satisfies $\mathbb{E}_Y\left[\sum_{l=1}^mg(a_{[l]}+Y)\right]\ge \bar{\Delta}_{\text{avg-opt}}$, then we have $z^\star_s=0$.

\end{proposition}
\begin{proof}
See Appendix \ref{Appendix_E}.
\end{proof}
\begin{algorithm}[!t]
\SetKwData{NULL}{NULL}
\textbf{given} $l=0$, sufficiently large $u$, tolerance $\epsilon_1>0$, tolerance $\epsilon_2>0$\;
\While{$u-l>\epsilon_1$}{
$\beta=\frac{l+u}{2}$\;
$J(\mathbf{s})=0$, $h(\mathbf{s})=0$, $h_{\text{last}}(\mathbf{s})=0$ for all states $\mathbf{s}\in\mathcal{S}$\;
\While{$\max_{\mathbf{s}\in\mathcal{S}}\vert h(\mathbf{s})-h_{\text{last}}(\mathbf{s})\vert> \epsilon_2$} {
\For{\textbf{each} $\mathbf{s}\in\mathcal{S}$}{
\uIf{$\mathbb{E}_Y\left[\sum_{l=1}^mg(a_{[l]}+Y)\right]\geq \beta$}{
$z^\star_s=0$\;}
\Else{
$z^\star_s=\text{argmin}_{z\in\mathcal{Z}}C(\mathbf{s},z)+\sum_{\mathbf{s}'\in\mathcal{S}}\mathbb{P}_{\mathbf{s}\mathbf{s}'}(z)h(\mathbf{s}')$\;
}
$J(\mathbf{s})=C(\mathbf{s},z^\star_s)+\sum_{\mathbf{s}'\in\mathcal{S}}\mathbb{P}_{\mathbf{s}\mathbf{s}'}(z^\star_s)h(\mathbf{s}')$\;
}
$h_{\text{last}}(\mathbf{s})=h(\mathbf{s})$\;
$h(\mathbf{s})=J(\mathbf{s})-J(\mathbf{o})$\;
}
\uIf{$J(\mathbf{o})\geq 0$}{
$u=\beta$\;}
\Else{
$l=\beta$\;
}
}
\caption{ RVI algorithm with reduced complexity.}\label{alg1}
\end{algorithm}
We can exploit the threshold test in Proposition \ref{th_thm} to reduce the complexity of the RVI algorithm as follows: The optimal waiting time for any state $\mathbf{s}$ that satisfies $\mathbb{E}_Y\left[\sum_{l=1}^mg(a_{[l]}+Y)\right]\geq \bar{\Delta}_{\text{avg-opt}}$ is zero. Thus, we need to solve \eqref{RVI1} only for the states that fail this threshold test. As a result, we reduce the number of computations required along the system state space, which reduces the complexity of the RVI algorithm. Note that $\bar{\Delta}_{\text{avg-opt}}$ can be obtained using the bisection method or any other one-dimensional search method. Combining this with the result of Proposition \ref{th_thm} and the RVI algorithm,  we propose the ``RVI with reduced complexity (RVI-RC) sampler'' in Algorithm \ref{alg1}. In the outer  layer of Algorithm \ref{alg1},  bisection is employed to obtain $\bar{\Delta}_{\text{avg-opt}}$, where $\beta$ converges to $\bar{\Delta}_{\text{avg-opt}}$.

Note that, according to  \cite{puterman2005markov,kaelbling1996recent}, $J(\mathbf{o})$ in Algorithm \ref{alg1} converges to the optimal average cost per stage. Moreover, the value of $u$ in Algorithm \ref{alg1} can be initialized to the value of the Ta-AP of the zero-wait sampler (as the Ta-AP of the zero-wait sampler provides an upper bound on the optimal Ta-AP), which can be easily calculated. 

The RVI algorithm and Whittle's methodology have been used in literature to obtain the optimal age scheduler in time-slotted multi-source networks (e.g.,\cite{hsu2017scheduling,hsu2018age}). Since they considered a time-slotted system, their model belongs to the class of Markov decision problems. In contrast, we consider random discrete transmission times that can be more than one time slot. Thus, our model belongs to the class of  semi-Markov decision problems, and hence is different from those in \cite{hsu2017scheduling,hsu2018age}.

In conclusion, an optimal solution for Problem \eqref{optimal_eq} is manifested in the following theorem:
\begin{theorem}\label{thm_taa}
The MAF scheduler and the RVI-RC sampler form an optimal solution for Problem \eqref{optimal_eq}. 
\end{theorem}
\begin{proof}
The theorem follows directly from Proposition \ref{Thm1}, Proposition \ref{thm2}, and Proposition \ref{th_thm}.
\end{proof}

\subsubsection{Special Case of $g(x)=x$}
Now we consider the case of $g(x)=x$ and obtain some useful insights. Define $A_s=\sum_{l=1}^ma_{[l]}$ as the sum of the age values of state $\mathbf{s}$. The threshold test in Proposition \ref{th_thm} is simplified as follows:
\begin{proposition}\label{prop_9_special_case}
If the state $\mathbf{s}=(a_{[1]}, \ldots, a_{[m]})$ satisfies $A_{s}\geq (\bar{\Delta}_{\text{avg-opt}}-m\mathbb{E}[Y])$, then we have $z^\star_s=0$.
\end{proposition}
\begin{proof}
The proposition follows directly by substituting $g(x)=x$ into the threshold test in Proposition \ref{th_thm}.
\end{proof}
Hence, the only change in Algorithm \ref{alg1} is to replace the threshold test in Step 7 by $A_{s}\geq (\bar{\Delta}_{\text{avg-opt}}-m\mathbb{E}[Y])$. Let $y_{\text{inf}}=\inf \{y\in\mathcal{Y}: \mathbb{P}[Y=y]>0\}$, i.e., $y_{\text{inf}}$ is the smallest possible transmission time in $\mathcal{Y}$. As a result of Proposition \ref{prop_9_special_case}, we obtain the following sufficient condition for the optimality of the zero-wait sampler for minimizing the Ta-AP when $g(x)=x$:
\begin{theorem}\label{zero_wait_samp_optimalty}
If 
\begin{align}\label{zero_wait_cond_32}
y_{\text{inf}}\geq \frac{(m-1)\mathbb{E}[Y]^2+\mathbb{E}[Y^2]}{(m+1)\mathbb{E}[Y]},
\end{align}
then the zero-wait sampler is optimal for Problem \eqref{equivilent_optimal_sampler}.
\end{theorem}
\begin{proof}
See Appendix \ref{Appendix_E'}
\end{proof}
From Theorem \ref{zero_wait_samp_optimalty}, it immediately follows that:
\begin{corollary}\label{zero_opt_time_slotted}
If the transmission times are positive and constant (i.e., $Y_i=const>0$ for all $i$), then the zero-wait sampler is optimal for Problem \eqref{equivilent_optimal_sampler}.
\end{corollary}
\begin{proof}
The corollary follows directly from Theorem \ref{zero_wait_samp_optimalty} by showing that \eqref{zero_wait_cond_32} always holds in this case.
\end{proof}
Corollary \ref{zero_opt_time_slotted} suggests that the designed schedulers in
\cite{aphermedis_he2017optimal,kadota2016minimizing,kadota2016minimizing_journal,hsu2017scheduling,kadota2018optimizing,hsu2018age,yates2017status,talak2018distributed,talak2018optimizing,talak2018scheduling,talak2018optimizing2} are indeed optimal in time-slotted systems. However, if there is a variation in the transmission times, these schedulers alone may not be optimal anymore, and we need to optimize the sampling times as well. 

\section{Low-complexity Sampler Design via Bellman's Equation Approximation}\label{Bellman}
In this section, we try to obtain low-complexity samplers via an approximate analysis for Bellman's equation in \eqref{bell1'}. The obtained low-complexity samplers in this section will be shown to have near optimal age performance in our numerical results in Section \ref{Simulations}. For a given state $\mathbf{s}$, we denote the next state given $z$ and $y$ by $\mathbf{s}'(z,y)$. We can observe that the transition probability in \eqref{trans_prob_eq} depends only on the distribution of the packet service time which is independent of the system state and the control action. Hence, the second term in Bellman's equation in \eqref{bell1'} can be rewritten as
\begin{align}
\sum_{\mathbf{s}'\in\mathcal{S}}\mathbb{P}_{\mathbf{s}\mathbf{s}'}(z)h(\mathbf{s}'(z,y))=\sum_{y\in\mathcal{Y}}\mathbb{P}(Y=y)h(\mathbf{s}'(z,y)).
\end{align}
As a result, Bellman's equation in \eqref{bell1'} can be rewritten as 
\begin{equation}\label{bell2'}
\!\!\!\!\lambda=\min_{z}\!\left( \!C(\mathbf{s},z)\!+\sum_{y\in\mathcal{Y}}\mathbb{P}(Y=y)(h(\mathbf{s}'(z,y))\!-\!h(\mathbf{s}))\!\right).
\end{equation}
Although $h(\mathbf{s})$ is discrete, we can interpolate the value of $h(\mathbf{s})$ between the discrete values so that it is differentiable by following the same approach in \cite{bettesh2006optimal} and \cite{wangdelay}. Let $\mathbf{s}=(a_{[1]},\ldots,a_{[m]})$, then using the first order Taylor approximation around a state $\mathbf{v}=(a_{[1]}^v,\ldots,a_{[m]}^v)$ (some fixed state), we get
\begin{equation}\label{taylor_exp1}
h(\mathbf{s})\approx h(\mathbf{v})+\sum_{l=1}^{m}(a_{[l]}-a_{[l]}^v)\frac{\partial h(\mathbf{v})}{\partial a_{[l]}}.
\end{equation}
Again, we use the first order Taylor approximation around the state $\mathbf{v}$, together with the state evolution in \eqref{state_evol}, to get
\begin{equation}\label{taylor_exp2}
\begin{split}
h(\mathbf{s}'(z,y))\approx h(\mathbf{v})+(y-a_{[m]}^v)\frac{\partial h(\mathbf{v})}{\partial a_{[m]}}+\sum_{l=1}^{m-1}(a_{[l+1]}-a_{[l]}^v+z+y)\frac{\partial h(\mathbf{v})}{\partial a_{[l]}}.
\end{split}
\end{equation}
From \eqref{taylor_exp1} and \eqref{taylor_exp2}, we get
\begin{equation}
\begin{split}
h(\mathbf{s}'(z,y))-h(\mathbf{s})\approx  (y-a_{[m]})\frac{\partial h(\mathbf{v})}{\partial a_{[m]}}+\sum_{l=1}^{m-1}(a_{[l+1]}-a_{[l]}+z+y)\frac{\partial h(\mathbf{v})}{\partial a_{[l]}}.
\end{split}
\end{equation}
This implies that
 \begin{equation}\label{taylor1}
\begin{split}
\sum_{y\in\mathcal{Y}}\mathbb{P}(Y=y)(h(\mathbf{s}'(z,y))-h(\mathbf{s})) \approx(\mathbb{E}[Y]-a_{[m]})\frac{\partial h(\mathbf{v})}{\partial a_{[m]}}+\sum_{l=1}^{m-1}(a_{[l+1]}-a_{[l]}+z+\mathbb{E}[Y])\frac{\partial h(\mathbf{v})}{\partial a_{[l]}}.
\end{split}
\end{equation}
Using \eqref{bell2'} with \eqref{taylor1}, we can get the following approximated Bellman's equation: 
\begin{equation}\label{approx1}
\begin{split}
\lambda\approx\min_z ~\bigg(C(\mathbf{s},z)+(\mathbb{E}[Y]-a_{[m]})\frac{\partial h(\mathbf{v})}{\partial a_{[m]}}+\sum_{l=1}^{m-1}(a_{[l+1]}-a_{[l]}+z+\mathbb{E}[Y])\frac{\partial h(\mathbf{v})}{\partial a_{[l]}}\bigg).
\end{split}
\end{equation}
By following the same steps as in Appendix \ref{Appendix_E} to get the optimal $z$ that minimizes the objective function in \eqref{approx1}, we get the following condition: The optimal $z$, for a given state $\mathbf{s}$, must satisfy
\begin{align}
&\mathbb{E}_Y\left[\sum_{l=1}^mg(a_{[l]}+t+Y)\right]-\bar{\Delta}_{\text{avg-opt}}+\sum_{l=1}^{m-1}\frac{\partial h(\mathbf{v})}{\partial a_{[l]}}\geq 0\label{eqder_bel_39}
\end{align}
for all $t>z$, and 
\begin{align}
\mathbb{E}_Y\left[\sum_{l=1}^mg(a_{[l]}+t+Y)\right]-\bar{\Delta}_{\text{avg-opt}}+\sum_{l=1}^{m-1}\frac{\partial h(\mathbf{v})}{\partial a_{[l]}}\leq 0\label{eqder_bel_40}
\end{align}
for all $t<z$. The smallest $z$ that satisfies \eqref{eqder_bel_39}-\eqref{eqder_bel_40} is 
\begin{equation}\label{eqder_bell_41}
\begin{split}
\hat{z}^\star_s=\inf\bigg\{ t\geq 0 : \mathbb{E}_Y\left[\sum_{l=1}^m g(a_{[l]}+t+Y)\right]
\geq\bar{\Delta}_{\text{avg-opt}}-\sum_{l=1}^{m-1}\frac{\partial h(\mathbf{v})}{\partial a_{[l]}}\bigg\},
\end{split}
\end{equation}
where $\hat{z}^*_s$ is the optimal solution of the approximated Bellman's equation for state $\mathbf{s}$. Note that the term  $\sum_{i=1}^{m-1}\frac{\partial h(\mathbf{v})}{\partial a_{[i]}}$ is constant. Hence, \eqref{eqder_bell_41} can be rewritten as
\begin{equation}\label{eqder_bell_42}
\hat{z}^\star_s=\inf\left\lbrace t\geq 0 : \mathbb{E}_Y\left[\sum_{l=1}^mg(a_{[l]}+t+Y)\right]
\geq T\right\rbrace. 
\end{equation}
This simple threshold sampler can approximate the optimal sampler for the original Bellman's equation in \eqref{bell1'}. The optimal threshold ($T$) in \eqref{eqder_bell_42} can be obtained using a golden-section method \cite{press1992golden}. Moreover, for a given state $\mathbf{s}$ and the threshold ($T$), \eqref{eqder_bell_42} can be solved using the bisection method or any other one-dimensional search method.

\textbf{Low-complexity Water-filling Sampler}:
Consider the case that $g(x)=x$, the solution in \eqref{eqder_bell_42} can be further simplified. Substituting $g(x)=x$ into \eqref{eqder_bell_41}, where the equality holds in this case, we get the following condition: The optimal $z$ in this case, for a given state $\mathbf{s}$, must satisfy
\begin{equation}\label{opt}
A_s-\bar{\Delta}_{\text{avg-opt}}+mz+m\mathbb{E}[Y]+\sum_{l=1}^{m-1}\frac{\partial h(\mathbf{v})}{\partial a_{[l]}}=0,
\end{equation}
where $A_s$ is the sum of the age values of state $\mathbf{s}$.
Rearranging \eqref{opt}, we get
\begin{equation}\label{optimal_sol1}
\hat{z}^*_s=\left[\frac{\bar{\Delta}_{\text{avg-opt}}-m\mathbb{E}[Y]-\sum_{l=1}^{m-1}\frac{\partial h(\mathbf{v})}{\partial a_{[l]}}}{m}-\frac{A_s}{m} \right]^+.
\end{equation}
By observing that the term  $\sum_{i=1}^{m-1}\frac{\partial h(\mathbf{v})}{\partial a_{[i]}}$ is constant, \eqref{optimal_sol1} can be rewritten as
\begin{equation}\label{fin_approx_sol}
\hat{z}^*_s=\left[T-\frac{A_s}{m} \right]^+,
\end{equation}
 The solution in \eqref{fin_approx_sol} is in the form of the water-filling solution as we compare a fixed threshold ($T$) with the average age of a state $\mathbf{s}$. The solution in \eqref{fin_approx_sol} suggests that this simple water-filling sampler can approximate the optimal solution of the original Bellman's equation in \eqref{bell1'} when $g(x)=x$. Similar to the general case, the optimal threshold ($T$) in \eqref{fin_approx_sol} can be obtained using a golden-section method. We evaluate the performance of the 
approximated samplers in the next section. 

\section{Numerical Results}\label{Simulations}
We present numerical results to evaluate our proposed solutions. We consider an information update system with $m=3$ sources. We use ``RAND" to represent a random scheduler, where sources are chosen to be served with equal probability. By  ``Constant-wait", we refer to the sampler that imposes a constant waiting time after each delivery with $Z_i=0.3\mathbb{E}[Y],~\forall i$. Moreover, we use ``Threshold'' and ``Water-filling'' to denote the proposed samplers in \eqref{eqder_bell_42} and \eqref{fin_approx_sol}, respectively.

\begin{figure}[t]
\centering
\includegraphics[scale=0.4]{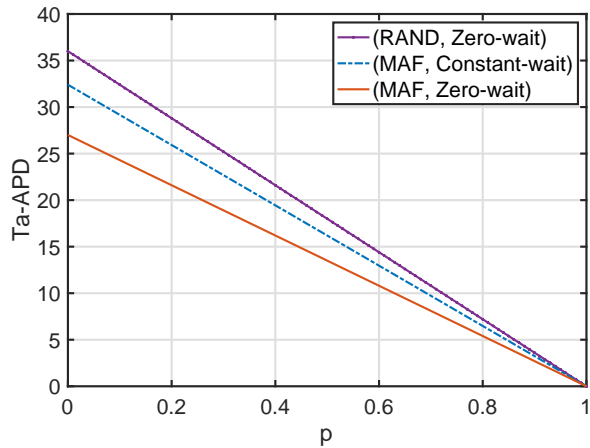}
\caption{Ta-APD versus the probability $p$ for an update system with $m=3$ sources, where $g(x)=x$.}
\label{peak_age_sim}
\end{figure}
We set the transmission times to be either 0 or 3 with probability $p$ and $1-p$, respectively. Fig. \ref{peak_age_sim} illustrates the Ta-APD versus the probability $p$, where we have $g(x)=x$. As we can observe, with fixing the sampler to the zero-wait one,   the MAF scheduler provides a lower Ta-APD than that of the RAND scheduler. Moreover, with fixing the scheduling policy to the MAF scheduler, the zero-wait sampler provides a lower Ta-APD compared to the constant-wait sampler. These observations agree with Theorem \ref{thm_tapa}. However, as we will see later, zero-wait sampler does not always minimize the Ta-AP. 

\begin{figure}[t]
\centering
\includegraphics[scale=0.4]{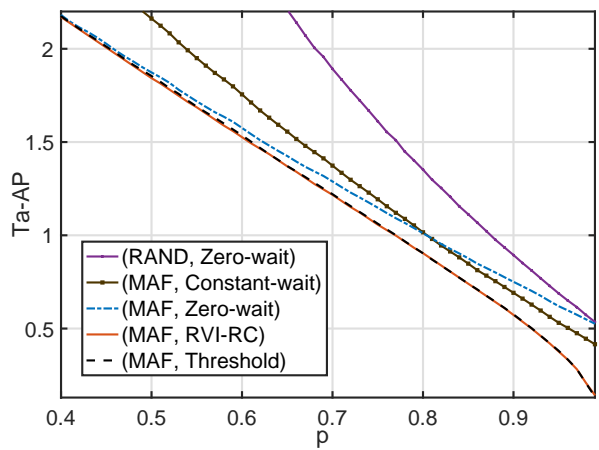}
\caption{Ta-AP versus the probability $p$ for an update system with $m=3$ sources, where $g(x)=e^{0.1x}-1$.}
\label{avg_age_penalty_expo_service_prob}
\end{figure}
\begin{figure}[t]
\centering
\includegraphics[scale=0.4]{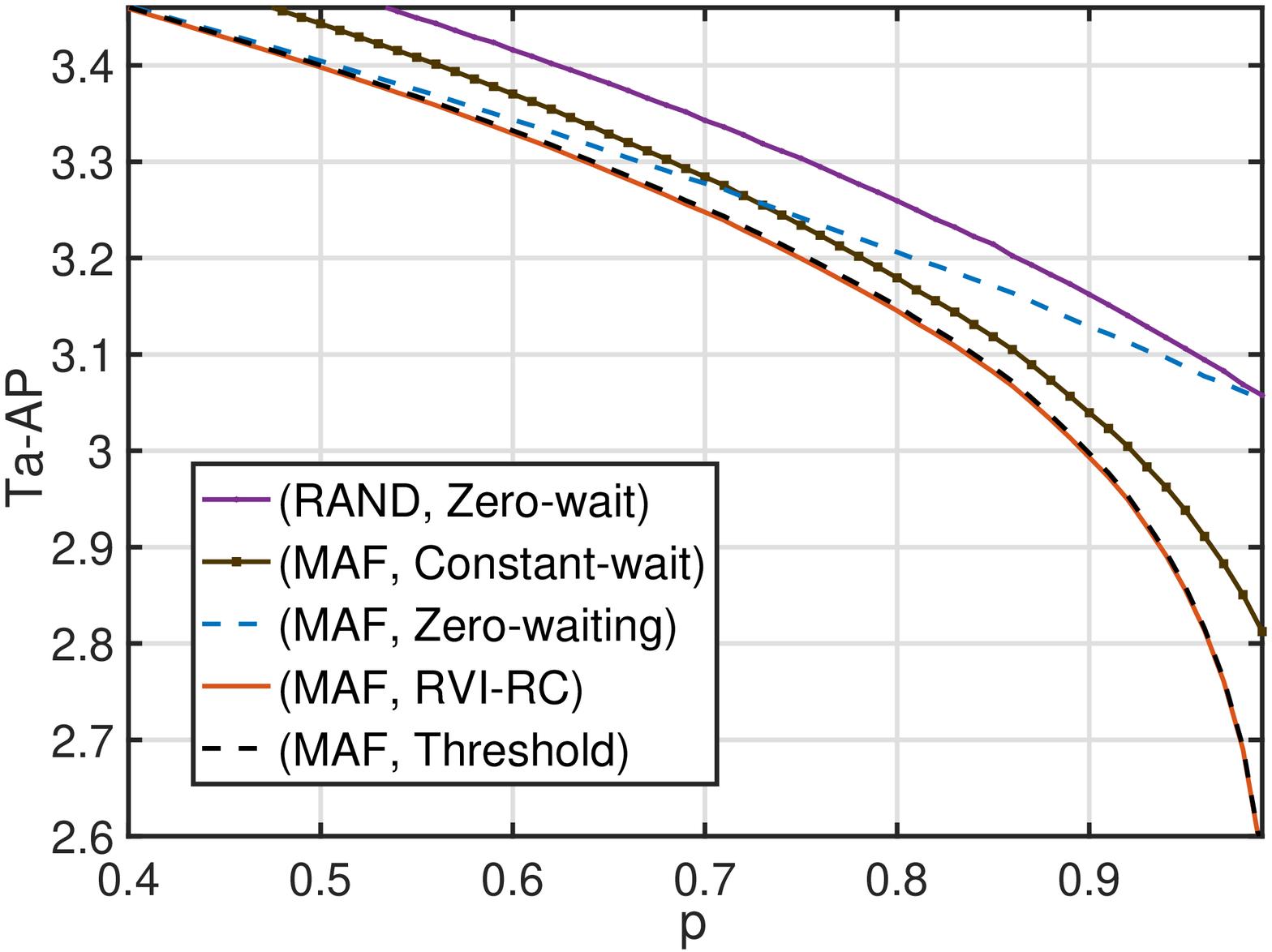}
\caption{Ta-AP versus the probability $p$ for an update system with $m=3$ sources, where $g(x)=x^{0.1}$.}
\label{avg_age_penalty_power_service_prob}
\end{figure}
\begin{figure}[t]
\centering
\includegraphics[scale=0.4]{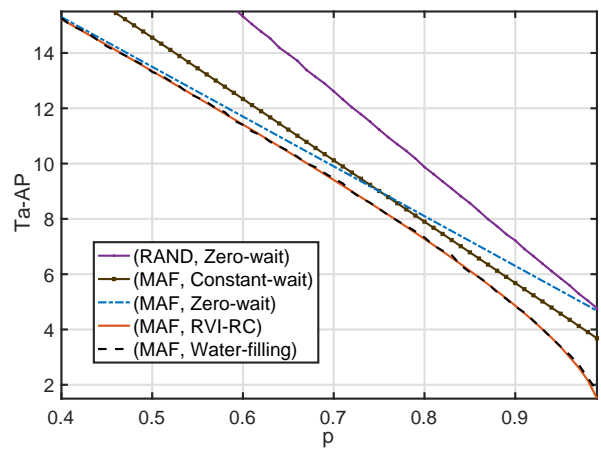}
\caption{Ta-AP versus the probability $p$ for an update system with $m=3$ sources, where $g(x)=x$.}
\label{avg_age_sim}
\end{figure}
We now evaluate the performance of our proposed solutions for minimizing the Ta-AP. We set the transmission times to be either 0 or 3 with probability $p$ and $1-p$, respectively. Figs. \ref{avg_age_penalty_expo_service_prob}, \ref{avg_age_penalty_power_service_prob}, and \ref{avg_age_sim} illustrate the Ta-AP versus the probability $p$, where we set the age-penalty function $g(x)$ to be $e^{0.1x}-1$, $x^{0.1}$, and $x$, respectively. The range of the probability $p$ is $[0.4, 0.99]$ in Figs. \ref{avg_age_penalty_expo_service_prob}, \ref{avg_age_penalty_power_service_prob}, and \ref{avg_age_sim}. When $p=1$, $\mathbb{E}[Y] = \mathbb{E}[Y^2] = 0$ and hence the Ta-AP of the zero-wait sampler (for any scheduler) at $p=1$ is undefined. Therefore, the point $p=1$ is not plotted in Figs. \ref{avg_age_penalty_expo_service_prob}, \ref{avg_age_penalty_power_service_prob}, and \ref{avg_age_sim}. For the zero-wait sampler, we find that the MAF scheduler provides a lower Ta-AP than that of the RAND scheduler. This agrees with Proposition \ref{Thm1}. Moreover, when the scheduling policy is fixed to the MAF scheduler, we find that the Ta-AP resulting from the RVI-RC sampler is lower than those resulting from the zero-wait sampler and the constant-wait sampler. This observation suggests the following: i) The zero-wait sampler does not necessarily minimize the Ta-AP, ii) optimizing the scheduling policy only is not enough to minimize the Ta-AP, but we have to optimize both the scheduling policy and the sampling policy together to minimize the Ta-AP. 
 In addition, as we can observe, the Ta-AP resulting from the threshold sampler in Figs. \ref{avg_age_penalty_expo_service_prob} and \ref{avg_age_penalty_power_service_prob}, and the  water-filling sampler in Fig. \ref{avg_age_sim} almost coincides with the Ta-AP resulting from the RVI-RC sampler. 

\begin{figure}[t]
\centering
\includegraphics[scale=0.4]{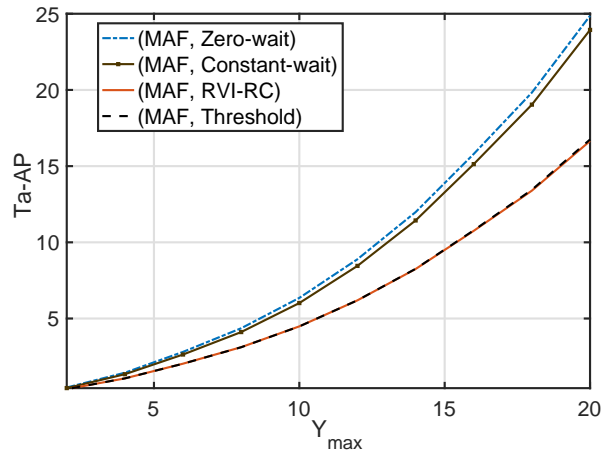}
\caption{Ta-AP versus the maximum service time $Y_{\text{max}}$ for an update system with $m=3$ sources, where $g(x)=e^{0.1x}-1$.}
\label{avg_age_penalty_expo_service_time}
\end{figure}
\begin{figure}[t]
\centering
\includegraphics[scale=0.4]{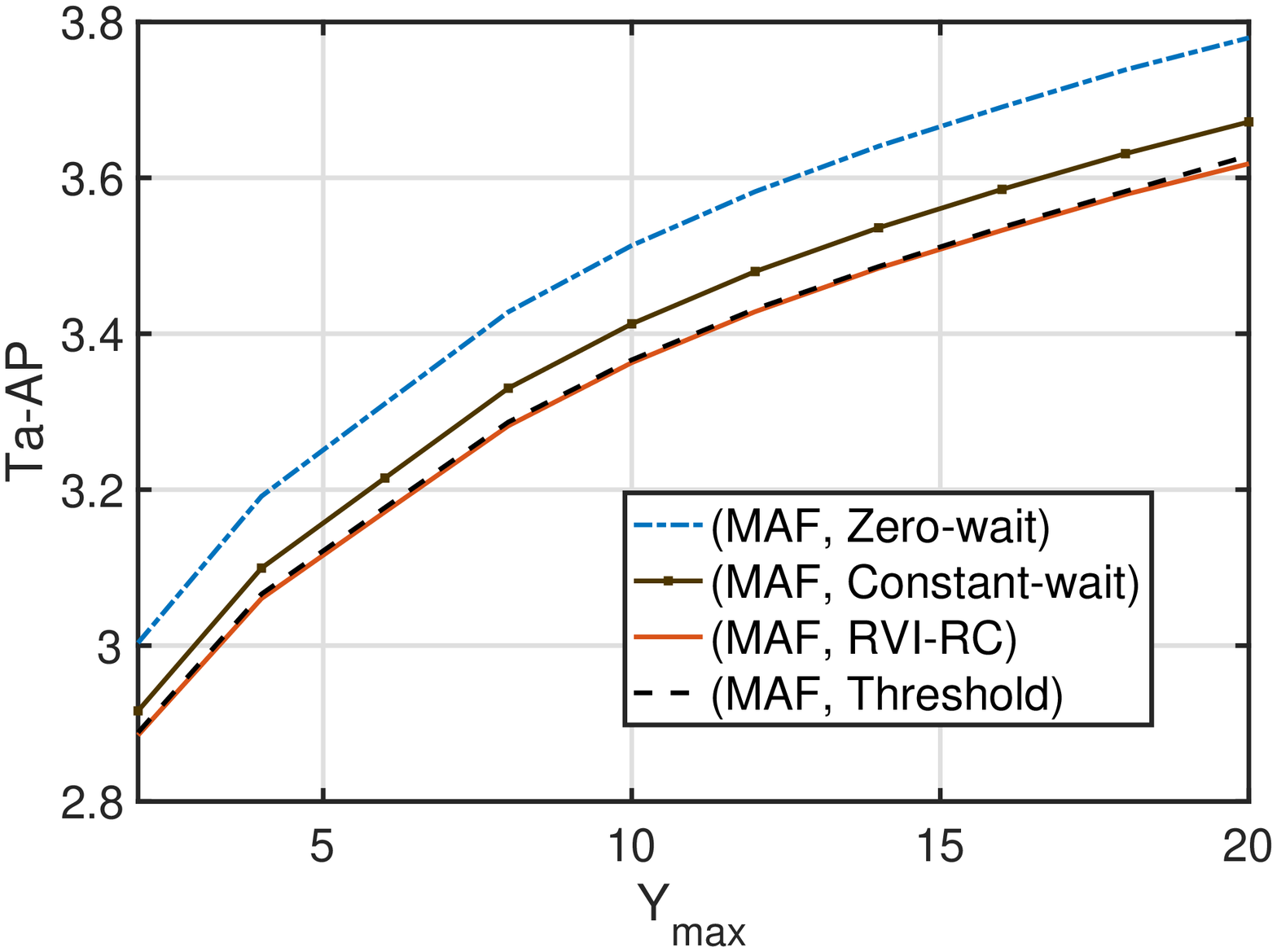}
\caption{Ta-AP versus the maximum service time $Y_{\text{max}}$ for an update system with $m=3$ sources, where $g(x)=x^{0.1}$.}
\label{avg_age_penalty_power_service_time}
\end{figure}
\begin{figure}[t]
\centering
\includegraphics[scale=0.4]{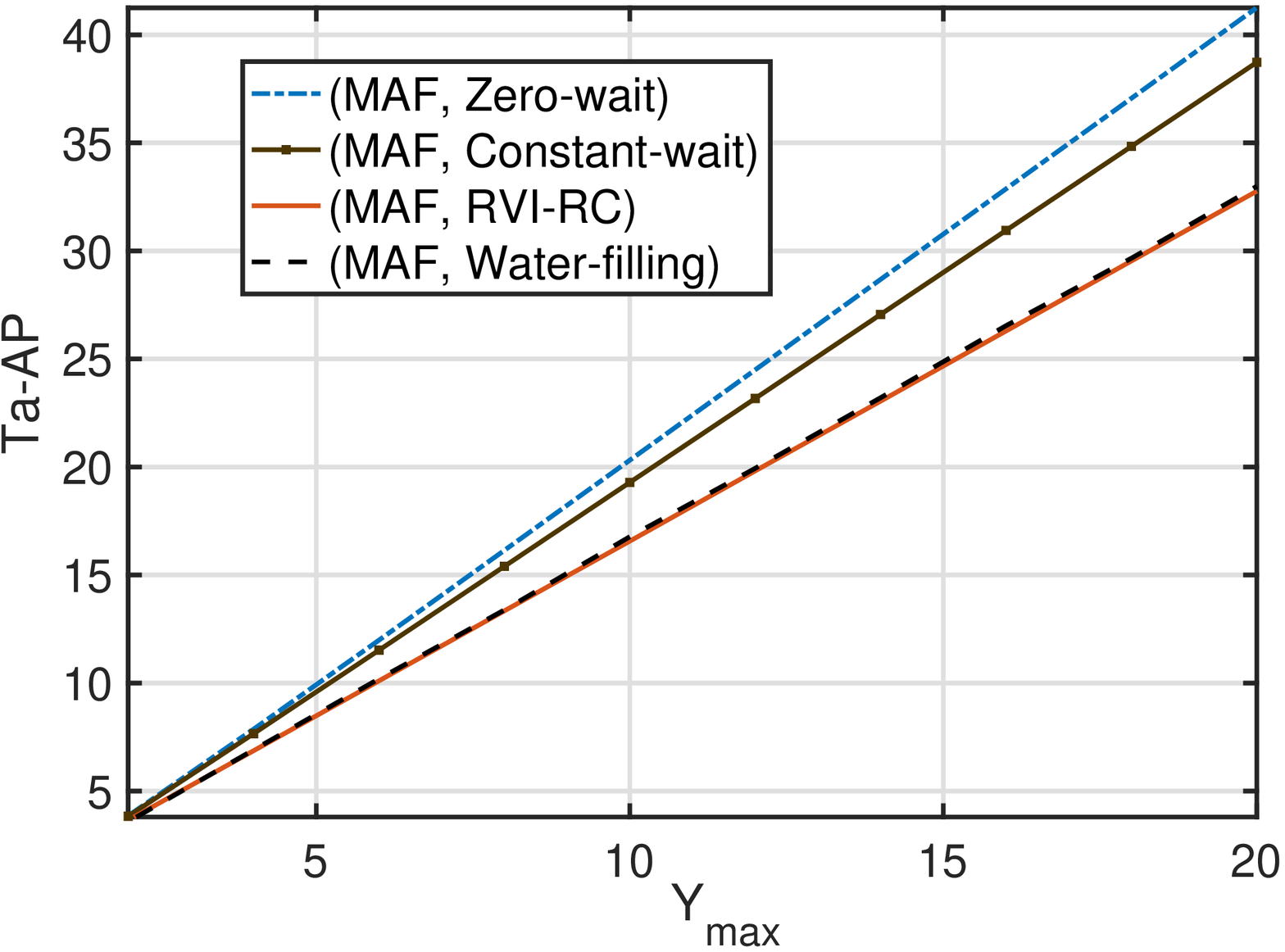}
\caption{Ta-AP versus the maximum service time $Y_{\text{max}}$ for an update system with $m=3$ sources, where $g(x)=x$.}
\label{avg_age_penalty_service_time}
\end{figure}
We then set the transmission times to be either 0 or $Y_{\text{max}}$ with probability $0.9$ and $0.1$, respectively. We vary the maximum transmission time $Y_{\text{max}}$ and plot the Ta-AP in Figs.   \ref{avg_age_penalty_expo_service_time},  \ref{avg_age_penalty_power_service_time}, and \ref{avg_age_penalty_service_time}, where $g(x)$ is set to be  $e^{0.1x}-1$, $x^{0.1}$, and $x$, respectively. The scheduling policy is fixed to the MAF scheduler in all plotted curves. We can observe in all figures that the Ta-AP resulting from the RVI-RC sampler is lower than those resulting from the zero-wait sampler and the constant-wait sampler, and the gap between them increases as the variability (variance) of the transmission times increases. This suggests that when the transmission times have a big variation, we have to optimize the scheduler and the sampler together to minimize the Ta-AP. We also can observe that the Ta-AP of the  threshold sampler in Figs. \ref{avg_age_penalty_expo_service_time} and \ref{avg_age_penalty_power_service_time}, and the water-filling sampler in Fig. \ref{avg_age_penalty_service_time} almost coincides with that of the RVI-RC sampler.

\begin{figure}[t]
\centering
\includegraphics[scale=0.4]{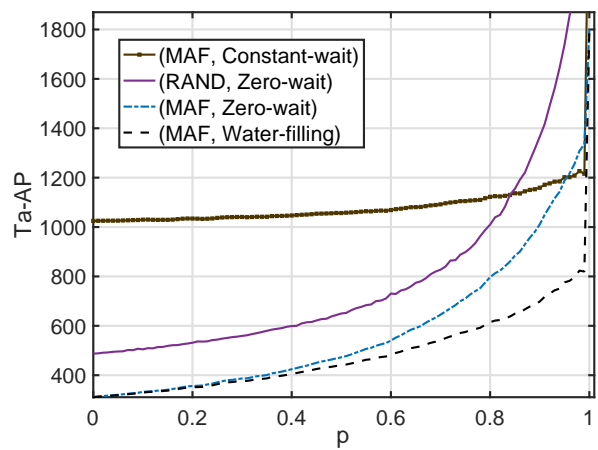}
\caption{Ta-AP versus the parameter $\sigma$ of the transmission time Markov chain for an update system with $m=10$ sources, where $g(x)=x$.}
\label{avg_age_penalty_service_time_correlated}
\end{figure}

Finally, we consider a larger scale update system with $m=10$. We model the transmission time as a discrete Markov chain with a probability mass function $\mathbb{P}[Y_i=1]=0.9$ and $\mathbb{P}[Y_i=30]=0.1$, and  a transition matrix
\begin{equation}
\begin{bmatrix}
\frac{8}{9}+\frac{\sigma}{9} & 1-\frac{8}{9}-\frac{\sigma}{9}\\
1-\sigma & \sigma
\end{bmatrix}.
\end{equation}
Fig. \ref{avg_age_penalty_service_time_correlated} illustrates the Ta-AP versus the transition matrix parameter $\sigma$, where $g(x)=x$. As we can observe, the (MAF, water-filling) policy provides the lowest Ta-AP compared to all plotted policies. Also, when $\sigma=1$, the transmission time reduces to be a constant time. We can observe that, when the scheduling policy is the MAF,  the Ta-APs achieved by the zero-wait and water-filling samplers are equal. This agrees with Corollary \ref{zero_opt_time_slotted}.
\section{Conclusion}\label{Concl}
In this work, we studied the problem of finding the optimal decision policy that controls the packet generation times and transmission order of the sources to minimize the Ta-APD and Ta-AP in a multi-source information update system. We showed that the MAF scheduler and the zero-wait sampler are jointly optimal for minimizing the Ta-APD. Moreover, we showed that the MAF scheduler and the RVI-RC sampler, which results from reducing the computation complexity of the RVI algorithm, are jointly optimal for minimizing the Ta-AP. Finally, we devised a low-complexity threshold sampler via an approximate analysis of Bellman's equation. This threshold sampler is further simplified to a simple water-filling sampler in the special case of linear age-penalty function. The numerical results showed that the performance of these approximated samplers is almost the same as that of the RVI-RC sampler.

\section{Appendix}
\appendices
\section{Proof of Proposition \ref{Thm1}}\label{Appendix_A}
We will need the following definitions:  A set $U\subseteq \mathbb{R}^n$ is called upper if $\mathbf{y}\in U$ whenever $\mathbf{y}\geq\mathbf{x}$ and $\mathbf{x}\in U$.
\begin{definition} \textbf{ Univariate Stochastic Ordering:} \cite{shaked2007stochastic} Let $X$ and $Y$ be two random variables. Then, $X$ is said to be stochastically smaller than $Y$ (denoted as $X\leq_{\text{st}}Y$), if
\begin{equation*}
\begin{split}
\mathbb{P}\{X>x\}\leq \mathbb{P}\{Y>x\}, \quad \forall  x\in \mathbb{R}.
 \end{split}
\end{equation*}
\end{definition}
\begin{definition}\label{def_2} \textbf{Multivariate Stochastic Ordering:} \cite{shaked2007stochastic} 
Let $\mathbf{X}$ and $\mathbf{Y}$ be two random vectors. Then, $\mathbf{X}$ is said to be stochastically smaller than $\mathbf{Y}$ (denoted as $\mathbf{X}\leq_\text{st}\mathbf{Y}$), if
\begin{equation*}
\begin{split}
\mathbb{P}\{\mathbf{X}\in U\}\leq \mathbb{P}\{\mathbf{Y}\in U\}, \quad \text{for all upper sets} \quad U\subseteq \mathbb{R}^n.
 \end{split}
\end{equation*}
\end{definition}
\begin{definition} \textbf{ Stochastic Ordering of Stochastic Processes:} \cite{shaked2007stochastic} Let $\{X(t), t\in [0,\infty)\}$ and $\{Y(t), t\in[0,\infty)\}$ be two stochastic processes. Then, $\{X(t), t\in [0,\infty)\}$ is said to be stochastically smaller than $\{Y(t), t\in [0,\infty)\}$ (denoted by $\{X(t), t\in [0,\infty)\}\leq_\text{st}\{Y(t), t\in [0,\infty)\}$), if, for all choices of an integer $n$ and $t_1<t_2<\ldots<t_n$ in $[0,\infty)$, it holds that
\begin{align}\label{law9'}
\!\!\!(X(t_1),X(t_2),\ldots,X(t_n))\!\leq_\text{st}\!(Y(t_1),Y(t_2),\ldots,Y(t_n)),\!\!
\end{align}
where the multivariate stochastic ordering in \eqref{law9'} was defined in Definition \ref{def_2}.
\end{definition}

Now, we prove Proposition \ref{Thm1}. Let the vector $\mathbf{\Delta}_\pi(t)=(\Delta_{[1],\pi}(t),\ldots, \Delta_{[m],\pi}(t))$ denote the system
state at time $t$ of the scheduler $\pi$,  where $\Delta_{[l],\pi}(t)$ is the $l$-th largest age of the sources at time $t$ under the scheduler $\pi$. Let $\{\mathbf{\Delta}_\pi(t), t\geq 0\}$ denote the state process of the scheduler $\pi$. 
For notational simplicity, let  $P$ represent the MAF scheduler. Throughout the proof, we assume that $\mathbf{\Delta}_\pi(0^-)=\mathbf{\Delta}_P(0^-)$ for all $\pi$ and the sampler is fixed to an arbitrarily chosen one. The key step in the proof of Proposition \ref{Thm1} is the following lemma, where we compare the scheduler $P$ with any arbitrary scheduler $\pi$.

\begin{lemma}\label{lem1thm1}
Suppose that $\mathbf{\Delta}_\pi(0^-)=\mathbf{\Delta}_P(0^-)$ for all scheduler $\pi$ and the sampler is fixed, then we have 
\begin{equation}\label{lem1thm1eq1}
\{\mathbf{\Delta}_P(t), t\geq 0\}\leq_{\text{st}}\{\mathbf{\Delta}_\pi(t), t\geq 0\}
\end{equation}
\end{lemma}
We use a coupling and forward induction to prove Lemma \ref{lem1thm1}. For any scheduler $\pi$, suppose that the stochastic processes $\widetilde{\mathbf{\Delta}}_P(t)$ and $\widetilde{\mathbf{\Delta}}_\pi(t)$ have the same stochastic laws as $\mathbf{\Delta}_P(t)$ and $\mathbf{\Delta}_\pi(t)$. The state processes $\widetilde{\mathbf{\Delta}}_P(t)$ and $\widetilde{\mathbf{\Delta}}_\pi(t)$ are coupled such that the packet service times are equal under both scheduling policies, i.e., $Y_i$'s are the same under both scheduling policies. Such a coupling is valid since the service time distribution is fixed under all policies. Since the sampler is fixed, such a coupling implies that the packet generation and delivery times are the same under both schedulers. According to Theorem 6.B.30 of \cite{shaked2007stochastic}, if we can show 
\begin{equation}\label{Thm1eq1}
\mathbb{P}\left[\widetilde{\mathbf{\Delta}}_P(t)\leq\widetilde{\mathbf{\Delta}}_\pi(t), t\geq 0\right]=1,
\end{equation}
then \eqref{lem1thm1eq1} is proven. To ease the notational burden, we will omit the tildes on the coupled versions in this proof and just use $\mathbf{\Delta}_P(t)$ and $\mathbf{\Delta}_\pi(t)$. Next, we compare scheduler $P$ and scheduler $\pi$ on a sample path and prove \eqref{lem1thm1eq1} using the following lemma:
\begin{lemma}[Inductive Comparison]\label{lem2thm1}
Suppose that a packet with generation time $S$ is delivered under the scheduler $P$ and the scheduler $\pi$ at the same time $t$. The system state of the scheduler $P$ is $\mathbf{\Delta}_P$ before the packet delivery, which becomes $\mathbf{\Delta}'_P$ after the packet delivery. The system state of the scheduler $\pi$ is $\mathbf{\Delta}_\pi$ before the packet delivery, which becomes $\mathbf{\Delta}'_\pi$ after the packet delivery. If
\begin{equation}\label{lem2thm1eq1}
\Delta_{[i],P}\leq \Delta_{[i],\pi}, i=1,\ldots,m,
\end{equation}
then
\begin{equation}\label{lem2thm1eq2}
\Delta_{[i],P}'\leq \Delta_{[i],\pi}', i=1,\ldots,m.
\end{equation}
\end{lemma}
Lemma \ref{lem2thm1}  is proven by following the proof idea of \cite[Lemma 2]{Yin_multiple_flows}. For the sake of completeness, we provide proof of Lemma \ref{lem2thm1}  as follows: 
\begin{proof}
Since only one source can be scheduled at a time and the scheduler $P$ is the MAF one, the packet with generation time $S$ must be generated from the source with maximum age $\Delta_{[1],P}$, call it source $l^*$. In other words, the age of source $l^*$ is reduced from the maximum age $\Delta_{[1],P}$ to the minimum age $\Delta_{[m],P}'=t-S$, and the age of the other $(m-1)$ sources	remain unchanged. Hence, 
\begin{equation}\label{prlem2thm1eq1}
\begin{split}
&\Delta_{[i],P}'=\Delta_{[i+1],P}, i=1,\ldots,m-1,\\
&\Delta_{[m],P}'=t-S.
\end{split}
\end{equation}
In the scheduler $\pi$, this packet can be generated from any source. Thus, for all cases of scheduler $\pi$, it must hold that
\begin{equation}\label{prlem2thm1eq2}
\Delta_{[i],\pi}'\geq\Delta_{[i+1],\pi}, i=1,\ldots,m-1.
\end{equation}
By combining \eqref{lem2thm1eq1}, \eqref{prlem2thm1eq1}, and \eqref{prlem2thm1eq2}, we have
\begin{equation}\label{prlem2thm1eq3}
\Delta_{[i],\pi}'\geq\Delta_{[i+1],\pi}\geq\Delta_{[i+1],P}=\Delta_{[i],P}', i=1,\ldots,m-1.
\end{equation}
In addition, since the same packet is also delivered under the scheduler $\pi$, the source from which this packet is generated under policy $\pi$ will have the minimum age after the delivery, i.e., we have
\begin{equation}\label{prlem2thm1eq4}
\Delta_{[m],\pi}'=t-S=\Delta_{[m],P}'.
\end{equation}
By this, \eqref{lem2thm1eq2} is proven. 
\end{proof}
\begin{proof}[Proof of Lemma \ref{lem1thm1}]
Using the coupling between the system state processes, and for any given sample path of the packet service times, we consider two cases:

\emph{Case 1:} When there is no packet delivery, the age of each source grows linearly with a slope 1.

\emph{Case 2:} When a packet is delivered, the ages of the sources evolve according to Lemma \ref{lem2thm1}.

By induction over time, we obtain
\begin{equation}
\Delta_{[i],P}(t)\leq \Delta_{[i],\pi}(t), i=1,\ldots,m, t\geq 0.
\end{equation}
Hence, \eqref{Thm1eq1} follows which implies \eqref{lem1thm1eq1} by Theorem 6.B.30 of \cite{shaked2007stochastic}. This completes the proof.
\end{proof}
\begin{proof}[Proof of Proposition \ref{Thm1}]
Since the Ta-APD and Ta-AP for any scheduling policy $\pi$ are the expectation of non-decreasing functional of the process $\{\mathbf{\Delta}_\pi(t), t\geq 0\}$, \eqref{lem1thm1eq1} implies \eqref{Thm1eq1b} and \eqref{Thm1eq2b} using the properties of stochastic ordering \cite{shaked2007stochastic}. This completes the proof. 
\end{proof}

\ifreport 
\section{Proof of Theorem \ref{thm_tapa}}\label{Appendix_A'}
The optimality of the MAF scheduler follows from Proposition \ref{Thm1}. Now, we need to show the optimality of the zero-wait sampler. We need to show that the Ta-APD is an increasing function of the packets waiting times $Z_i$'s. Define $K_{li}$ as the number of packets  that have been transmitted since the last received service by source $l$ before time $D_i$.  Also, let $\gamma_l$ be the index of the first delivered packet from source $l$. 

For $i>\gamma_l$, the last service that source $l$ has received before time $D_i^-$ was at time $D_{i-K_{li}}$. Since the age process increases linearly with time when there is no packet delivery, we have
\begin{equation}
\Delta_{l}(D_i^{-})=D_i-D_{i-K_{li}}+Y_{i-K_{li}}, ~i>\gamma_l,
\end{equation}
where  $Y_{i-K_{li}}$ is the service time of packet $i-K_{li}$. Note that $Y_{i-K_{li}}$ is also the age value of source $l$ at time $D_{i-K_{li}}$, i.e., $\Delta_{l}(D_{i-K_{li}})=Y_{i-K_{li}}$. Note that $D_{i}=Y_{i}+Z_{i-1}+D_{i-1}$. Repeating this, we can express $(D_{i}-D_{i-K_{li}})$ in terms of $Z_i$'s and $Y_i$'s, and hence we get
\begin{equation}\label{ari_interms_ziandyi}
\Delta_{l}(D_i^{-})=\sum_{k=0}^{K_{li}}Y_{i-k}+\sum_{k=1}^{K_{li}}Z_{i-k}, ~i>\gamma_l.
\end{equation}
For example, in Fig. \ref{fig:ages_evolv}, we have $\Delta_{2}(D_2^{-})=Y_1+Z_1+Y_2$.

For $i \leq\gamma_l$, $\Delta_{l}(D_i^{-})$ is simply the initial age value of source $l$ ($\Delta_{l}(0)$) plus the length of the time interval $[0,D_i)$. Hence, we have
\begin{equation}\label{ari_interms_ziandyi_2}
\Delta_{l}(D_i^{-})=\Delta_{l}(0)+D_i,~ i\leq\gamma_l.
\end{equation}
Again using $D_{i}=Y_{i}+Z_{i-1}+D_{i-1}$ and the fact that $D_0=0$, we get
\begin{equation}\label{ari_interms_ziandyi_2}
\Delta_{l}(D_i^{-})=\Delta_{l}(0)+\sum_{k=1}^{i}Y_{k}+\sum_{k=0}^{i}Z_{k}, ~ i\leq\gamma_l.
\end{equation}
In Fig. \ref{fig:ages_evolv}, For example, we have $\Delta_{1}(D_1^{-})=\Delta_{1}(0)+Z_0+Y_1$. 

Substituting \eqref{ari_interms_ziandyi} and \eqref{ari_interms_ziandyi_2} into \eqref{peak_age_def}, we get
\begin{equation}\label{peak_age_inc_eq1}
\begin{split}
\Delta_{\text{avg-D}}(\pi,f)=\limsup_{n\rightarrow\infty}\frac{1}{n}\mathbb{E}\Biggl[\sum_{l=1}^m\sum_{i=1}^{\gamma_l}g\left(\Delta_{l}(0)+\sum_{k=1}^{i}Y_{k}+\sum_{k=0}^{i}Z_{k}\right)+\sum_{i=\gamma_l}^{n}g\left(\sum_{k=0}^{K_{li}}Y_{i-k}+\sum_{k=1}^{K_{li}}Z_{i-k}\right)\Biggr].
\end{split}
\end{equation}
Since the function $g(\cdot)$ is non-decreasing,  \eqref{peak_age_inc_eq1} implies that the Ta-APD is a non-decreasing function of the waiting times. This completes the proof. \qed

\section{Proof of Lemma \ref{pro_simple}}\label{Appendix_B}
Part (i) is proven in two steps:

\emph{Step 1:} We will prove that $\bar{\Delta}_{\text{avg-opt}}\leq \beta$ if and only if $\Theta(\beta)\leq 0$. If $\bar{\Delta}_{\text{avg-opt}}\leq \beta$, there exists a sampling policy $f=(Z_0,Z_1,\ldots)\in\mathcal{F}$ that is feasible for \eqref{optimal_eq_sampler} and \eqref{equivilent_optimal_sampler}, which satisfies
\begin{equation}\label{lem3_simp_eq1}
\limsup_{n\rightarrow\infty}\frac{\sum_{i=0}^{n-1}\mathbb{E}\left[\sum_{l=1}^m\int_{a_{li}}^{a_{li}+Z_i+Y_{i+1}}g(\tau)d\tau\right]}{\sum_{i=0}^{n-1}\mathbb{E}[Z_i+Y_{i+1}]}\leq \beta.
\end{equation}
Hence,
\begin{equation}\label{lem3_simp_eq2}
\!\limsup_{n\rightarrow\infty}\frac{\!\!\frac{1}{n}\!\sum_{i=0}^{n\!-\!1}\!\mathbb{E}\!\left[\sum_{l=1}^m\!\int_{a_{li}}^{a_{li}\!+\!Z_i\!+\!Y_{i\!+\!1}}\!\!g(\tau)d\tau\!-\!\beta(Z_i\!+\!Y_{i\!+\!1})\right]\!}{\!\frac{1}{n}\sum_{i=0}^{n\!-\!1}\mathbb{E}[Z_i\!+\!Y_{i\!+\!1}]}\!\leq\! 0.
\end{equation} 
Since $Z_i$'s and $Y_i$'s are bounded and positive and $\mathbb{E}[Y_i]>0$ for all $i$, we have $0<\liminf_{n\rightarrow\infty}$ $\frac{1}{n}\sum_{i=0}^{n-1}\mathbb{E}[Z_i+Y_{i+1}]\leq \limsup_{n\rightarrow\infty}$ $\frac{1}{n}\sum_{i=0}^{n-1}\mathbb{E}[Z_i+Y_{i+1}]\leq q$ for some $q\in\mathbb{R}^{+}$. By this, we get
\begin{equation}\label{lem3_simp_eq3}
\limsup_{n\rightarrow\infty}\frac{1}{n}\!\sum_{i=0}^{n\!-\!1}\!\mathbb{E}\!\left[\sum_{l=1}^m\!\int_{a_{li}}^{a_{li}\!+\!Z_i\!+\!Y_{i\!+\!1}}\!\!g(\tau)d\tau\!-\!\beta(Z_i\!+\!Y_{i\!+\!1})\right]\leq 0.
\end{equation}
Therefore, $\Theta(\beta)\leq 0$.

In the reverse direction, if $\Theta(\beta)\leq 0$, then  there exists a sampling policy $f=(Z_0,Z_1,\ldots)\in\mathcal{F}$ that is feasible for \eqref{optimal_eq_sampler} and \eqref{equivilent_optimal_sampler}, which satisfies \eqref{lem3_simp_eq3}. Since we have $0<\liminf_{n\rightarrow\infty}$ $\frac{1}{n}\sum_{i=0}^{n-1}\mathbb{E}[Z_i+Y_{i+1}]\leq\limsup_{n\rightarrow\infty}$ $\frac{1}{n}\sum_{i=0}^{n-1}\mathbb{E}[Z_i+Y_{i+1}]\leq q$, we can divide \eqref{lem3_simp_eq3} by $\liminf_{n\rightarrow\infty}\frac{1}{n}$ $\sum_{i=0}^{n-1}\mathbb{E}[Z_i+Y_{i+1}]$ to get \eqref{lem3_simp_eq2}, which implies \eqref{lem3_simp_eq1}. Hence, $\bar{\Delta}_{\text{avg-opt}}\leq \beta$. By this, we have proven that  $\bar{\Delta}_{\text{avg-opt}}\leq \beta$ if and only if $\Theta(\beta)\leq 0$.

\emph{Step 2:} We need to prove that $\bar{\Delta}_{\text{avg-opt}}< \beta$ if and only if $\Theta(\beta)< 0$. This statement can be proven by using the arguments in Step 1, in which ``$\leq$" should be replaced by ``$<$''. Finally, from the statement of Step 1, it immediately follows that $\bar{\Delta}_{\text{avg-opt}}> \beta$ if and only if $\Theta(\beta)> 0$. This completes part (i).

Part(ii): We first show that each optimal solution to \eqref{optimal_eq_sampler} is an optimal solution to \eqref{equivilent_optimal_sampler}. 
By the claim of part (i), $\Theta(\beta)=0$ is equivalent to $\bar{\Delta}_{\text{avg-opt}}=\beta$. Suppose that policy $f=(Z_0,Z_1,\ldots)\in\mathcal{F}$ is an optimal solution to \eqref{optimal_eq_sampler}. Then, $\Delta_{\text{avg}(\pi_{\text{MAF}},f)}=\bar{\Delta}_{\text{avg-opt}}=\beta$. Applying this in the arguments of \eqref{lem3_simp_eq1}-\eqref{lem3_simp_eq3}, we can show that policy $f$ satisfies
\begin{equation}
\limsup_{n\rightarrow\infty}\frac{1}{n}\!\sum_{i=0}^{n\!-\!1}\!\mathbb{E}\!\left[\sum_{l=1}^m\!\int_{a_{li}}^{a_{li}\!+\!Z_i\!+\!Y_{i\!+\!1}}\!\!g(\tau)d\tau\!-\!\beta(Z_i\!+\!Y_{i\!+\!1})\right]= 0.
\end{equation}
This and $\Theta(\beta)=0$ imply that policy $f$ is an optimal solution to \eqref{equivilent_optimal_sampler}.

Similarly, we can prove that each optimal solution to \eqref{equivilent_optimal_sampler} is an optimal solution to \eqref{optimal_eq_sampler}. By this, part (ii) is proven.\qed

\section{Proof of Proposition \ref{thm2}}\label{Appendix_C}
According to \cite[Proposition 4.2.1 and Proposition 4.2.6]{Bertsekas1996bookDPVol2}, it is enough to show that for every two states $\mathbf{s}$ and $\mathbf{s}'$, there exists a stationary deterministic policy $f$ such that for some $k$, we have
\begin{equation}\label{cond_eq_req1}
\mathbb{P}\left[\mathbf{s}(k)=\mathbf{s}'\vert\mathbf{s}(0)=\mathbf{s}, f\right] > 0.
\end{equation}
From the state evolution equation \eqref{state_evol}, we can observe that any state in $\mathcal{S}$ can be represented in terms of the waiting and service times. This implies \eqref{cond_eq_req1}. To clarify this, let us consider a system with 3 sources. Assume that the elements of state $\mathbf{s}'$ are as follows:
\begin{equation}
\begin{split}
&a_{[1]}'=y_3+z_2+y_2+z_1+y_1,\\
&a_{[2]}'=y_3+z_2+y_2,\\
&a_{[3]}'=y_3,
\end{split}
\end{equation}
where $y_i$'s and $z_i$'s are any arbitrary elements in $\mathcal{Y}$ and $\mathcal{Z}$, respectively. Then, we will show that from any arbitrary state $\mathbf{s}=(a_{[1]},a_{[2]},a_{[3]})$, a sequence of service and waiting times can be followed to reach state $\mathbf{s}'$. If we have $Z_0=z_1$, $Y_1=y_1$, $Z_1=z_1$, $Y_2=y_2$, $Z_2=z_2$, and $Y_3=y_3$, then according to \eqref{state_evol}, we have in the first stage
\begin{equation}
\begin{split}
&a_{[1]1}=a_{[2]}+z_1+y_1,\\
&a_{[2]1}=a_{[3]}+z_1+y_1,\\
&a_{[3]1}=y_1,
\end{split}
\end{equation}
and in the second stage, we have
\begin{equation}
\begin{split}
&a_{[1]2}=a_{[3]}+z_1+y_2+z_1+y_1,\\
&a_{[2]2}=y_2+z_1+y_1,\\
&a_{[3]2}=y_2,
\end{split}
\end{equation}
and in the third stage, we have
\begin{equation}
\begin{split}
&a_{[1]3}=y_3+z_2+y_2+z_1+y_1=a_{[1]}',\\
&a_{[2]3}=y_3+z_2+y_2=a_{[2]}',\\
&a_{[3]3}=y_3=a_{[3]}'.
\end{split}
\end{equation}
Hence, a stationary deterministic policy $f$ can be designed to reach state $\mathbf{s}'$ from state $\mathbf{s}$ in 3 stages, if the aforementioned sequence of service times occurs. This implies that
\begin{equation}
\mathbb{P}\left[\mathbf{s}(3)=\mathbf{s}'\vert\mathbf{s}(0)=\mathbf{s}, f\right] =\prod_{i=1}^3 \mathbb{P}(Y_i=y_i)>0,
\end{equation}
where we have used that $Y_i$'s are \emph{i.i.d.}\footnote{We assume that all elements in $\mathcal{Y}$ have a strictly positive probability, where the elements with zero probability can be removed without affecting the proof.} The previous argument can be generalized to any number of sources. In particular, a forward induction over $m$ can be used to show the result, where \eqref{cond_eq_req1} trivially holds for $m=1$, and the previous argument can be used to show that \eqref{cond_eq_req1} holds for any general $m$. This completes the proof. \qed

\section{Proof of Proposition \ref{th_thm}}\label{Appendix_E}
We prove Proposition \ref{th_thm} into two steps:

\textbf{Step 1}: We first show that $h(\mathbf{s})$ is non-decreasing in $\mathbf{s}$. 
To do so, we show that $J_\alpha(\mathbf{s})$, defined in \eqref{j_alpha}, is non-decreasing in $\mathbf{s}$, which together with \eqref{relative_cost_eq} imply that $h(\mathbf{s})$ is non-decreasing in $\mathbf{s}$.

Given an initial state $\mathbf{s}(0)$, the total expected discounted cost under a sampling policy $f\in\mathcal{F}$ is given by
\begin{equation}
J_\alpha(\mathbf{s}(0); f)=\limsup_{n\rightarrow\infty}\mathbb{E}\left[ \sum_{i=0}^{n-1}\alpha^iC(\mathbf{s}(i), Z_i)\right],
\end{equation}
where $0<\alpha<1$ is the discount factor. The optimal total expected $\alpha$-discounted cost function is defined by
\begin{equation}
J_\alpha(\mathbf{s})=\min_{f\in\mathcal{F}}J_\alpha(\mathbf{s}; f),~\mathbf{s}\in\mathcal{S}.
\end{equation}
A policy is said to be $\alpha$-optimal if it minimizes the total expected $\alpha$-discounted cost. The discounted cost optimality equation of $J_\alpha(\mathbf{s})$ is discussed below.
\begin{proposition}
The optimal total expected $\alpha$-discounted cost $J_\alpha(\mathbf{s})$ satisfies
\begin{equation}\label{optimality_cond}
J_\alpha(\mathbf{s})=\min_{z\in\mathcal{Z}} C(\mathbf{s}, z)+\alpha\sum_{\mathbf{s}'\in\mathcal{S}}\mathbb{P}_{\mathbf{s}\mathbf{s}'}(z)J_\alpha(\mathbf{s}').
\end{equation}
Moreover, a stationary deterministic policy that attains the minimum in equation \eqref{optimality_cond} for each $\mathbf{s}\in\mathcal{S}$ will be an $\alpha$-optimal policy. Also, let $J_{\alpha,0}(\mathbf{s})=0$ for all $\mathbf{s}$ and any $n\geq 0$,
\begin{equation}\label{value_iter_disc}
J_{\alpha,n+1}(\mathbf{s})=\min_{z\in\mathcal{Z}} C(\mathbf{s}, z)+\alpha\sum_{\mathbf{s}'\in\mathcal{S}}\mathbb{P}_{\mathbf{s}\mathbf{s}'}(z)J_{\alpha,n}(\mathbf{s}').
\end{equation}
Then, we have $J_{\alpha,n}(\mathbf{s})\rightarrow J_{\alpha}(\mathbf{s})$ as $n\rightarrow\infty$ for every $\mathbf{s}$, and $\alpha$.
\end{proposition}
\begin{proof}
Since we have bounded cost per stage, the proposition follows directly from
\cite[Proposition 1.2.2 and Proposition 1.2.3]{Bertsekas1996bookDPVol2}, and \cite{sennott1989average}.
\end{proof}
Next, we use the optimality equation \eqref{optimality_cond} and the value iteration in \eqref{value_iter_disc} to prove that $J_\alpha(\mathbf{s})$ is non-decreasing in $\mathbf{s}$.
\begin{lemma}\label{lem2}
The optimal total expected $\alpha$-discounted cost function $J_\alpha(\mathbf{s})$ is non-decreasing in $\mathbf{s}$.
\end{lemma}
\begin{proof}
We use induction on $n$ in equation \eqref{value_iter_disc} to prove Lemma \ref{lem2}. Obviously, the result holds for $J_{\alpha,0}(\mathbf{s})$.

Now, assume that $J_{\alpha,n}(\mathbf{s})$ is non-decreasing in $\mathbf{s}$. We need to show that for any two states $\mathbf{s}_1$ and $\mathbf{s}_2$ with $\mathbf{s}_1\leq \mathbf{s}_2$, we have $J_{\alpha,n+1}(\mathbf{s}_1)\leq J_{\alpha,n+1}(\mathbf{s}_2)$. First, we note that, since the age-penalty function $g(\cdot)$ is non-decreasing, the expected cost per stage $C(\mathbf{s},z)$ is non-decreasing in $\mathbf{s}$, i.e., we have
\begin{equation}\label{pf1}
C(\mathbf{s}_1,z)\leq C(\mathbf{s}_2,z).
\end{equation}
From the state evolution equation \eqref{state_evol} and the transition probability equation \eqref{trans_prob_eq}, the second term of the right-hand side (RHS) of \eqref{value_iter_disc} can be rewritten as
\begin{equation}
\sum_{\mathbf{s}'\in\mathcal{S}}\mathbb{P}_{\mathbf{s}\mathbf{s}'}(z)J_{\alpha,n}(\mathbf{s}')=\sum_{y\in\mathcal{Y}}\mathbb{P}(Y=y)J_{\alpha,n}(\mathbf{s}'(z,y)),
\end{equation}
where $\mathbf{s}'(z,y)$ is the next state from state $\mathbf{s}$ given the values of $z$ and $y$. Also, according to the state evolution equation \eqref{state_evol}, if the next states of $\mathbf{s}_1$ and $\mathbf{s}_2$ for given values of $z$ and $y$ are $\mathbf{s}'_1(z,y)$ and $\mathbf{s}'_2(z,y)$, respectively, then we have $\mathbf{s}'_1(z,y)\leq \mathbf{s}'_2(z,y)$. This implies that
\begin{equation}\label{pf2}
\sum_{y\in\mathcal{Y}}\mathbb{P}(Y=y)J_{\alpha,n}(\mathbf{s}'_1(z,y))\leq \sum_{y\in\mathcal{Y}}\mathbb{P}(Y=y)J_{\alpha,n}(\mathbf{s}'_2(z,y)),
\end{equation}
where we have used the induction assumption that  $J_{\alpha,n}(\mathbf{s})$ is non-decreasing in $\mathbf{s}$. Using \eqref{pf1}, \eqref{pf2}, and the fact that the minimum operator in \eqref{value_iter_disc} retains the non-decreasing property, we conclude that 
\begin{equation}
J_{\alpha,n+1}(\mathbf{s}_1)\leq J_{\alpha,n+1}(\mathbf{s}_2).
\end{equation}
This completes the proof.
\end{proof}

\textbf{Step 2:} We use Step 1 to prove Proposition \ref{th_thm}. From Step 1, we have that $h(\mathbf{s})$ is non-decreasing in $\mathbf{s}$. Similar to Step 1, this implies that the second term of the right-hand side (RHS) of \eqref{bell1'} ($\sum_{\mathbf{s}'\in\mathcal{S}}\mathbb{P}_{\mathbf{s}\mathbf{s}'}(z)h(\mathbf{s}')$) is non-decreasing in $\mathbf{s}'$. Moreover, from the state evolution \eqref{state_evol}, we can notice that, for any state $\mathbf{s}$, the next state $\mathbf{s}'$ is increasing in $z$. This argument implies that the second term of the right-hand side (RHS) of \eqref{bell1'} ($\sum_{\mathbf{s}'\in\mathcal{S}}\mathbb{P}_{\mathbf{s}\mathbf{s}'}(z)h(\mathbf{s}')$) is increasing in $z$. Thus, the value of $z\in\mathcal{Z}$ that achieves the minimum value of this term is zero. If, for a given state $\mathbf{s}$,
the value of $z\in\mathcal{Z}$ that achieves the minimum value of the cost function $C(\mathbf{s},z)$ is zero, then $z=0$ solves the RHS of \eqref{bell1'}. In the sequel, we obtain the condition on $\mathbf{s}$ under which $z=0$ minimizes the cost function $C(\mathbf{s},z)$.

Now, we focus on the cost function $C(\mathbf{s},z)$.
In order to obtain the optimal $z$ that minimizes this cost function, we need to obtain the one-sided derivative of it. The one-sided derivative of a function $q$ in the direction of $\omega$ at $z$ is given by
\begin{equation}
\delta q(z;\omega)\triangleq \lim_{\epsilon\to 0^+}\frac{q(z+\epsilon\omega)-q(z)}{\epsilon}.
\end{equation}
Let $r(\mathbf{s}, z, Y)=\sum_{l=1}^m\int_{a_{[l]}}^{a_{[l]}+z+Y}g(\tau)d\tau$. 
Since $r(\mathbf{s}, z, Y)$ is the sum of integration of a non-decreasing function $g(\cdot)$, it is easy to show that $r(\mathbf{s}, z, Y)$ is convex. According to \cite[Lemma 4]{SunJournal2016}, the function $q(z)=\mathbb{E}_Y\left[r(\mathbf{s}, z, Y)\right]$ is convex as well. Hence, the one-sided derivative $\delta q(z;\omega)$ of $q(z)$ exists \cite[p.709]{Bertsekas}. Moreover, since $z\to r(\mathbf{s}, z, Y) $ is convex, the function $\epsilon\to [r(\mathbf{s}, z+\epsilon\omega, Y)-r(\mathbf{s}, z, Y)]/\epsilon$ is non-decreasing and bounded from above on $(0,\theta]$ for some $\theta>0$ \cite[Proposition 1.1.2(i)]{butnariu2000totally}. Using the monotone convergence theorem \cite[ Theorem 1.5.6]{durrett2010probability}, we can interchange the limit and integral operators in $\delta q(z;\omega)$ such that 
\begin{align}
\delta q(z;\omega)=&\lim_{\epsilon\to 0^+}\frac{1}{\epsilon}\mathbb{E}_Y[r(\mathbf{s}, z+\epsilon\omega, Y)-r(\mathbf{s}, z, Y)]\nonumber\\
=&~ \mathbb{E}_Y\left[\lim_{\epsilon\to 0^+}\frac{1}{\epsilon}\{r(\mathbf{s}, z+\epsilon\omega, Y)-r(\mathbf{s}, z, Y)\}\right]\nonumber\\
=&~ \mathbb{E}_Y\left[\lim_{t\to z^+}\sum_{l=1}^m g(a_{[l]}+t+Y)w\mathbbm{1}_{\{\omega>0\}}\right.\nonumber\left.+\lim_{t\to z^-}\sum_{l=1}^mg(a_{[l]}+t+Y)w\mathbbm{1}_{\{\omega<0\}}\right]
\nonumber
\\
=& \lim_{t\to z^+}\mathbb{E}_Y\left[\sum_{l=1}^mg(a_{[l]}+t+Y)w\mathbbm{1}_{\{\omega>0\}}\right]\nonumber+\lim_{t\to z^-}\mathbb{E}_Y\left[\sum_{l=1}^mg(a_{[l]}+t+Y)w\mathbbm{1}_{\{\omega<0\}}\right],
\end{align}
where $\mathbbm{1}_E$ is the indicator function of event $E$. According to  \cite[p.710]{Bertsekas} and the convexity of $q(z)$, $z$ is optimal to the cost function $C(\mathbf{s},z)$ if and only if
\begin{equation}\label{eqder_74}
\delta q(z;\omega)-\bar{\Delta}_{\text{avg-opt}}\omega\ge 0, ~\forall \omega\in\mathbb{R}.
\end{equation}
As $\omega$ in \eqref{eqder_74} is an arbitrary real number, considering $\omega=1$, \eqref{eqder_74} becomes
\begin{equation}\label{eqder_75}
\lim_{t\to z^+}\mathbb{E}_Y\left[\sum_{l=1}^mg(a_{[l]}+t+Y)\right]-\bar{\Delta}_{\text{avg-opt}}\geq 0.
\end{equation}
Likewise, considering $\omega=-1$, \eqref{eqder_74} implies
\begin{equation}\label{eqder_76}
\lim_{t\to z^-}\mathbb{E}_Y\left[\sum_{l=1}^mg(a_{[l]}+t+Y)\right]-\bar{\Delta}_{\text{avg-opt}}\leq 0.
\end{equation}
Since $g(\cdot)$ is non-decreasing, we get from \eqref{eqder_74}-\eqref{eqder_76} that $z$ must satisfy
\begin{align}
&\mathbb{E}_Y\left[\sum_{l=1}^mg(a_{[l]}+t+Y)\right]-\bar{\Delta}_{\text{avg-opt}}\geq 0,~ \text{if}~t>z,\label{eqder_77}\\&
\mathbb{E}_Y\left[\sum_{l=1}^mg(a_{[l]}+t+Y)\right]-\bar{\Delta}_{\text{avg-opt}}\leq 0,~ \text{if}~t<z.\label{eqder_78}
\end{align}
Subsequently, the smallest $z$ that satisfies \eqref{eqder_77}-\eqref{eqder_78} is
\begin{equation}\label{eqder_79}
\!\!\!\!z=\inf\left\lbrace t\geq 0 : \mathbb{E}_Y\left[\sum_{l=1}^mg(a_{[l]}+t+Y)\right]\geq\bar{\Delta}_{\text{avg-opt}}\right\rbrace. 
\end{equation}
According to \eqref{eqder_79}, Since $g(\cdot)$ is non-decreasing, if $\mathbb{E}_Y\left[\sum_{l=1}^mg(a_{[l]}+Y)\right]\geq\bar{\Delta}_{\text{avg-opt}}$, then $z=0$ minimizes $C(\mathbf{s},z)$. This completes the proof. \qed

\section{Proof of Theorem \ref{zero_wait_samp_optimalty}}\label{Appendix_E'}
We use the threshold test $A_s\geq (\bar{\Delta}_{\text{avg-opt}}-m\mathbb{E}[Y])$, in Proposition \ref{prop_9_special_case}, to prove Theorem \ref{zero_wait_samp_optimalty}.  We will show that the condition in \eqref{zero_wait_cond_32} implies that $A_s\geq (\bar{\Delta}_{\text{avg-opt}}-m\mathbb{E}[Y])$ holds for  all states $\mathbf{s}\in\mathcal{S}$, and hence the zero-wait sampler is optimal under this condition. From the state evolution \eqref{state_evol}, we can deduce that for any state $\mathbf{s}\in\mathcal{S}$, we have
\begin{align}
a_{[l]}\geq (m-l+1)y_{\text{inf}},~\forall l=1,\ldots,m.
\end{align}
This implies
\begin{align}
A_s\geq \sum_{l=1}^m ly_{\text{inf}}=\frac{m(m+1)}{2}y_{\text{inf}}, ~\forall \mathbf{s}\in\mathcal{S}.
\end{align}
Moreover, it is easy to show that the total-average age of the zero-wait sampler, when the scheduling policy is fixed to the MAF scheduler, is given by
\begin{align}
\bar{\Delta}_0=\frac{\frac{m(m+1)}{2}\mathbb{E}[Y]^2+\frac{m}{2}\mathbb{E}[Y^2]}{\mathbb{E}[Y]}.
\end{align}
Since $\bar{\Delta}_0\geq\bar{\Delta}_{\text{avg-opt}}$, we have
\begin{align}
\bar{\Delta}_0-m\mathbb{E}[Y]\ge\bar{\Delta}_{\text{avg-opt}}-m\mathbb{E}[Y].
\end{align}
Hence, if the following condition holds
\begin{align}
\frac{m(m+1)}{2}y_{\text{inf}}\geq \frac{\frac{m(m+1)}{2}\mathbb{E}[Y]^2+\frac{m}{2}\mathbb{E}[Y^2]}{\mathbb{E}[Y]}-m\mathbb{E}[Y],
\end{align}
which is equivalent to 
\begin{align}
y_{\text{inf}}\geq\frac{(m-1)\mathbb{E}[Y]^2+\mathbb{E}[Y^2]}{(m+1)\mathbb{E}[Y]},
\end{align}
then we have $A_s\geq (\bar{\Delta}_{\text{avg-opt}}-m\mathbb{E}[Y])$ for all states $\mathbf{s}\in\mathcal{S}$. This implies that the zero-wait sampler is optimal under this condition. This completes the proof. \qed
\fi
\bibliographystyle{IEEEbib}
\bibliography{MyLib}
\end{document}